\documentclass[aps,pra,twocolumn,groupedaddress,showkeys]{revtex4-1}
\usepackage{amsmath,paralist,amsthm,comment,amssymb}
\usepackage{times,color}
\usepackage{epsfig}
\usepackage{tikz}
\usepackage{pgfplots}
\usepackage{multirow}
\usepackage{subfigure}
 \usepackage{kbordermatrix} 
\usetikzlibrary{patterns} 

\usepackage{algorithm} 
\usepackage{algorithmic}

\DeclareMathOperator{\wt}{wt}

\DeclareMathOperator{\rk}{rank}

\newcommand{\mc}[1]{\mathcal{#1}}

\newcommand{\F}{\mathbb{F}}
\newcommand{\nix}[1]{}

\newtheorem{theorem}{Theorem}
\newtheorem{corollary}[theorem]{Corollary}
\newtheorem{lemma}{Lemma}

\newtheorem{remark}{Remark}

\renewcommand{\thetheorem}{\Alph{theorem}}

\newcommand{\ba}{\begin{array}}
\newcommand{\ea}{\end{array}}
\newcommand{\ben}{\begin{eqnarray}}
\newcommand{\een}{\end{eqnarray}}

\newcommand{\be}{\begin{eqnarray*}}
\newcommand{\ee}{\end{eqnarray*}}

\begin{document}


\title{ Topological Subsystem Codes From Graphs and Hypergraphs}


\author{Pradeep Sarvepalli}
\email[]{pradeep.sarvepalli@gatech.edu}
\author{Kenneth R. Brown}
\affiliation{Schools of Chemistry and Biochemistry; Computational Science and Engineering; and Physics, Georgia Institute of Technology, GA 30332}

\date{June 30, 2012}

\begin{abstract}
Topological subsystem codes were proposed by Bombin based on 3-face-colorable cubic graphs. Suchara, Bravyi and Terhal generalized this construction  and proposed a method to construct topological subsystem codes using 3-valent hypergraphs that satisfy certain constraints. Finding such hypergraphs and computing their parameters however is a  nontrivial task.  We propose families of topological subsystem codes that were previously not known.  In particular, our constructions give codes which cannot be derived from Bombin's construction.  We also study the error recovery schemes for the proposed subsystem codes and give detailed schedules for the syndrome measurement that take advantage of the 2-locality of the gauge group.  The study also leads to a new and general construction for color codes.
\end{abstract}

\pacs{}
\keywords{quantum codes, topological quantum codes, subsystem codes, hypergraphs, decoding, color codes}

\maketitle
\section{Introduction}

A major challenge in the theory of quantum error correcting codes is to design codes that are
well suited for fault tolerant quantum computing. Such codes have many stringent requirements 
imposed on them, constraints that are usually not considered in the design of classical codes.
An important metric that captures the suitability of a family of codes for 
fault tolerant quantum computing  is the threshold of that family of codes. 
Informally, the threshold of a sequence  of codes of increasing length is the maximum error rate that can be tolerated by the family as we increase the length of the codes in the sequence.
The threshold is affected by numerous factors and there is no single parameter that we can optimize to design codes with high threshold. Furthermore, in the literature thresholds are reported under 
various assumptions. As the authors of \cite{landahl11} noted, there are three thresholds that
are of interest: i) the code threshold which assumes there are no measurement errors, ii) the phenomenological threshold which incorporates to some extent the errors due to measurement errors, and iii) the circuit 
threshold which incorporates all errors due to gates and measurements. For a given family of codes, 
invariably the code threshold is the highest and the circuit threshold the lowest. 

One of the nonidealities that affects the lowering of thresholds is the introduction of measurement
errors. So codes which have same code thresholds, such as the toric codes and color codes,
can end up with different circuit thresholds \cite{landahl11,fowler11}. At this point one can attempt 
to  improve the circuit threshold by 
designing codes that have efficient recovery schemes and are more robust 
to measurement errors among other things. An important development in this direction has come in the
form of subsystem codes, also called as operator error correcting codes \cite{bacon06a,kribs05, kribs05b,kribs05c,poulin05,ps08}. By providing additional degrees
of freedom subsystem codes allow us to design recovery schemes which are 
more robust to circuit nonidealities. That they can improve the threshold has already been reported
in the literature \cite{aliferis06}.

A class of codes that have been found to be suitable for fault tolerant computing are the 
topological codes. These codes have local stabilizer generators, enabling the design of a local 
architecture for computing with them and also have the highest thresholds reported so far \cite{raussen07}. It is 
tempting to combine the benefits of these codes with the ideas of subsystem codes. This was first achieved in the work of Bombin \cite{bombin10}, followed by Suchara,  Bravyi  and Terhal \cite{suchara10}.

However, the code thresholds reported in \cite{suchara10} were lower than the
thresholds of the toric codes and color codes. Nonetheless, this should not lead us to a hasty conclusion that the topological subsystem codes are not as good as the toric codes. There are at least
two reasons why topological subsystem codes warrant further investigation. Firstly, the threshold reported in \cite{suchara10} is about 2\% while \cite{andrist12} showed that the topological subsystem
codes can have a threshold as high as 5.5\%.   This motivates the
further study on decoding topological subsystem codes that are closer to their theoretical limits
as well as the study of subsystem codes that have higher code thresholds. 

The second point that must be borne in mind is the rather surprising lower circuit threshold of color 
code on the square  octagon lattice as compared to the toric codes. Both of these codes have a code 
threshold of about 11\%. But the circuit threshold of the color codes is about an order of magnitude 
lower than that of  the toric codes. Both codes enable local architectures for fault tolerant quantum computing, both architectures realize gates by code 
deformation techniques, and both achieve universality in quantum computation through magic state 
distillation. Moreover, the color codes considered in \cite{landahl11} unlike the surface code can even 
realize the entire Clifford group transversally. Despite this apparent advantage over the toric codes, 
the color codes lose out to the surface codes in one crucial aspect---the weight of the check operators. 
Some of the check operators for the square octagon color code have a
weight that is twice the weight of the check operators in the toric codes. Even though these higher weight check operators are approximately a fifth of the operators, they appear to be the dominant reason for the lower circuit threshold of the color codes. 

The preceding discussion indicates that measurement errors can severely undermine the performance of a code with many good properties including a good code threshold.
Thus  any improvement in circuit techniques or error recovery schemes to make the circuits
more robust to these errors are likely to yield significant improvements in the circuit thresholds. 
This is precisely where topological subsystem codes come into picture. Because they can be designed to 
function with just two-body  measurements,  these codes can greatly mitigate the detrimental effects of
measurement errors. A  strong case  in favor of the suitability of the subsystem codes
with current quantum information technologies has  already been made in \cite{suchara10}. 

For all these reasons topological subsystem codes are worth further investigation. 
This work is aimed at realizing the potential of topological subsystem codes.  Our main contribution
in this paper is to give  large classes of topological subsystem codes, which were not previously
known in literature. 
Our results put at our disposal a huge arsenal of topological subsystem codes, which aids in the 
evaluation of their promise for fault tolerant quantum computing. In addition to building upon the 
work of \cite{suchara10} it also sheds light on color codes, an area of independent interest.

The paper is structured as follows. After reviewing the necessary background on subsystem codes
in Section~\ref{sec:bg},  we give our main results in Section~\ref{sec:tsc}. Then in 
Section~\ref{sec:decoding} we show how to measure the stabilizer for the proposed  codes in a consistent  fashion. 
We conclude with a brief discussion on the significance of these results in Section~\ref{sec:summary}.

\section{Background and Previous Work}\label{sec:bg}
\subsection{Subsystem codes}
In the standard model of quantum error-correction, information is protected by encoding it into a
subspace of the system Hilbert space. 
In the subsystem model of error correction \cite{bacon06a,kribs05, kribs05b,kribs05c,poulin05,ps08}, the  subspace is further decomposed as 
$L\otimes G$.  The subsystem $L$  encodes the logical information, while the subsystem $G$ provides 
additional degrees of freedom; it is also called the gauge subsystem and said to encode the gauge qubits.    The notation  $[[n,k,r,d]]$ is used to denote a subsystem code on $n$ qubits, with $\dim L= 2^k$ and $\dim G= 2^r$ and able to detect errors of weight up to $d-1$ on the 
subsystem $L$.  In this model  an $[[n,k,d]]$ quantum code is the same as an  $[[n,k,0,d]]$  subsystem code.

The introduction of the gauge subsystem allows us to simplify the error 
recovery schemes \cite{bacon06a,aliferis06} since errors that act only on the gauge subsystem need not be corrected.  Although sometimes this comes at the expense of a reduced encoding rate, nonetheless as in the case of the Bacon-Shor code, this can substantially improve the performance with respect to the corresponding stabilizer code associated with it without affecting the rate \cite{bacon06a}.

We assume that the reader is familiar with the stabilizer formalism for quantum codes \cite{calderbank98, gottesman97}. We briefly review it for the subsystem codes \cite{poulin05,ps08}. A subsystem code is defined by a (nonableian) subgroup of the Pauli group;
it is called the gauge group $\mc{G}$ of the subsystem code. We denote by $S'=Z(\mc{G})$, the centre of $\mc{G}$. Let $\langle i{\bf I}, S\rangle =S'$. The subsystem code is simply the space stabilized by $S$. (Henceforth, we shall ignore phase factors and  let $S$ be equivalent to $S'$.) 
Henceforth, we shall ignore the phase factors and let $S$.
The bare logical operators of the code are given by the elements in $C(\mc{G})$, the centralizer of $\mc{G}$. (We view the identity also as a logical operator.) These logical operators do not act on the gauge subsystem but only on the information subsystem. The operators in $C(S)$ are called dressed logical operators and in general they also act on the gauge subsystem as well.
For an $[[n,k,r,d]]$ subsystem code, with the stabilizer dimension $\dim S = s$, we have the following relations:
\ben
n&= & k+r+s,\\
\dim \mc{G} &= &2r+s,\label{eq:dim-G}\\
\dim C(\mc{G})  & = & 2k+s,\label{eq:dim-CG}\\
d & = & \min \{\wt(e) \mid e \in C(Z(\mc{G}))\setminus \mc{G} \}.\label{eq:distance}
\een
The notation $\wt(e)$ is used to denote the number of qubits on which the error $e$ acts nontrivially.

\subsection{Color codes}

In the discussion on topological codes, it is tacitly assumed that the code is associated to a graph 
which is embedded on some suitable surface. 
Color codes \cite{bombin06} are a class of topological codes derived from 3-valent graphs with the additional property that they are 3-face-colorable. Such graphs are called 2-colexes. The stabilizer of the color code
associated to such a 2-colex is generated by operators defined as follows:
\ben
B_{f}^{\sigma} = \prod_{i\in f} \sigma_i,  \sigma \in \{X,Z \},
\een
where $f$ is a face of $\Gamma_2$.
 A method to construct 2-colexes from standard graphs was proposed in \cite{bombin07b}. Because of its relevance for us we
briefly review it here. 
\renewcommand{\thealgorithm}{\Alph{algorithm}}

\renewcommand{\algorithmicrequire}{\textbf{Input:}}
\renewcommand{\algorithmicensure}{\textbf{Output:}}
\begin{algorithm}[H]
 \floatname{algorithm}{Construction}
\caption{{\ensuremath{\mbox{ Topological color code construction}}}}\label{proc:tcc-bombin}
\begin{algorithmic}[1]
\REQUIRE {An arbitrary graph $\Gamma$.}
\ENSURE {A 2-colex $\Gamma_2$.}
\STATE Color each face of the embedding by $x\in \{r,b,g\}$.
\STATE Split each edge into two edges and color the face by $y\in \{ r,b,g\}\setminus x$ as shown below.
\begin{center}
\begin{tikzpicture}
\draw [color=black,  thick](0,0)--(1,0);

\draw[fill =gray] (0,0) circle (2pt);
\draw[fill =gray] (1,0) circle (2pt);

\draw [color=black, fill=green!30, thick](2,0)..controls (2.5,0.25)..(3,0)--(3,0)..controls (2.5,-0.25)..(2,0);
\draw[fill =gray] (2,0) circle (2pt);
\draw[fill =gray] (3,0) circle (2pt);

\end{tikzpicture}
\end{center}

\STATE Transform  each vertex of degree $d$ into a face containing $d$  edges and color it 
$z\in  \{r,b,g \} \setminus \{ x,y\}$.
Denote this graph by $\Gamma_2$.
\begin{center}
\begin{tikzpicture}
\draw[fill=blue!30, color=blue!30] (1,0)--(-0.5,sin{60})--(-0.5,-sin{60});

\draw [color=black, fill=green!30, thick](0,0)..controls (0.5,0.25)..(1,0)--(1,0)..controls (0.5,-0.25)..(0,0);
\draw[fill =gray] (0,0) circle (2pt);
\draw[fill =gray] (1,0) circle (2pt);

\draw [color=black, fill=green!30, thick, rotate=120](0,0)..controls (0.5,0.25)..(1,0)--(1,0)..controls (0.5,-0.25)..(0,0);

\draw [color=black, fill=green!30, thick, rotate=-120](0,0)..controls (0.5,0.25)..(1,0)--(1,0)..controls (0.5,-0.25)..(0,0);

\draw[fill =gray] (0,0) circle (2pt);
\draw[fill =gray, rotate=120] (1,0) circle (2pt);
\draw[fill =gray,rotate=-120] (1,0) circle (2pt);

\draw [fill=red] (3,0)+(-30:0.5)--+(30:0.5)--+(90:0.5)--+(150:0.5)--+(210:0.5)--+(270:0.5)--+(-30:0.5);
\foreach \i in {0,120,...,300} 
{
\draw [color=green!30,fill=green!30] (3,0) +(\i-30:0.5)--+(\i-15:1) -- +(\i+15:1)-- +(\i+30:0.5)-- +(\i-30:0.5) ;
\draw [color=red, thick](3,0) +(\i+15:1)--+(\i+30:0.5);
\draw [color=red, thick](3,0) +(\i-30:0.5)--+(\i-30+15:1);
\draw [color=green, thick](3,0) +(\i+15:1)--+(\i+15:1.5);
\draw [color=green, thick](3,0) +(\i-30+15:1.5)--+(\i-30+15:1);
\draw[color=blue, thick] (3,0) +(\i+15:1)--+(\i-15:1);
\draw [fill  =gray] (3,0) +(\i+15:1) circle (2pt);
\draw [fill  =gray] (3,0) +(\i-15:1) circle (2pt);

}
\foreach \i in {0,120,...,240}
{
\draw [color=blue, thick](3,0) +(\i-30:0.5)--+(\i+30:0.5);
\draw [color=green, thick](3,0) +(\i+30:0.5)--+(\i+90:0.5);
}

\foreach \i in {0,60,...,300} 
{
\draw [fill  =gray] ((3,0) +(\i+30:0.5) circle (2pt);
}

\end{tikzpicture}
\end{center}
\end{algorithmic}
\end{algorithm}

Notice that in the above construction, every vertex, face and edge in $\Gamma$ lead to a face in 
$\Gamma_2$. Because of this correspondence, we shall call a face in $\Gamma_2$ a $v$-face if its parent
in $\Gamma$ was a vertex, a $f$-face if its parent was a face and an $e$-face if its parent 
was an edge. Note that an $e$-face is always 4-sided.

\subsection{Topological subsystem codes via color codes}

At the outset it is fitting to distinguish topological subsystem codes from non-topological codes
such as the Bacon-Shor codes that are nonetheless local. A more precise definition can be found in 
\cite{bravyi09,bombin10}, but for our purposes it suffices to state it in the following terms. 
\begin{compactenum}[(i)]
\item  The stabilizer  $S$ (and the gauge group) have local generators and $O(1)$ support. 
\item  Errors in  $C(S)$ that have a trivial homology on the surface are in the stabilizer,
while the undetectable errors have a nontrivial homology on the surface. 
\end{compactenum}

We denote the vertex set and edge set of a graph $\Gamma$ by $V(\Gamma)$, $E(\Gamma)$ respectively.
We denote the set of edges incident on a vertex $v$ by $\delta(v)$ and the edges that constitute the
boundary of a face by $\partial(f)$. We denote the Euler characteristic of a graph by $\chi$, where
$\chi=  |V(\Gamma)| -| E(\Gamma)|+|F(\Gamma)|$.
The dual of a graph is the graph obtained by replacing 
 every face $f$ with a vertex $f^\ast$, and for every edge in the boundary of two faces $f_1$
 and $f_2$, creating a dual edge connecting $f_1^\ast$ and $f_2^\ast$. 
The subsystem code construction due to \cite{bombin10} takes the dual of a  2-colex, 
 and modifies it to obtain a subsystem code. 
 The procedure is outlined below:

\begin{algorithm}[H]
 \floatname{algorithm}{Construction}
\caption{{\ensuremath{\mbox{ Topological subsystem code construction}}}}\label{proc:tsc-bombin}
\begin{algorithmic}[1]
\REQUIRE {An arbitrary 2-colex $\Gamma_2$.}
\ENSURE{Topological subsystem code. }
\STATE Take the dual of $\Gamma_2$. It is a 3-vertex-colorable graph. 
\STATE Orient each edge as a directed edge as per the following:

\begin{center}
\begin{tikzpicture}
\draw [color=black, ->, thick](0,0)--(1,0);
\draw[fill =red] (0,0) circle (2pt);
\draw[fill =blue] (1,0) circle (2pt);

\draw [color=black, ->, thick](2,0)--(3,0);
\draw[fill =blue] (2,0) circle (2pt);
\draw[fill =green] (3,0) circle (2pt);

\draw [color=black, ->, thick](4,0)--(5,0);
\draw[fill =green] (4,0) circle (2pt);
\draw[fill =red] (5,0) circle (2pt);

\end{tikzpicture}
\end{center}

\STATE Transform  each (directed) edge into a 4-sided face.
\begin{center}
\begin{tikzpicture}
\draw [color=black,  ->, thick](0,0)--(1,0);
\draw[fill =gray] (0,0) circle (2pt);
\draw[fill =black] (1,0) circle (2pt);

\draw [color=blue, thick] (1.5,0.5)--(2.5,0.5); 
\draw [color=blue, thick] (1.5,-0.5)--(2.5,-0.5);
\draw[color=red, ultra thick] (1.5,0.5)--(1.5,-0.5);
\draw[color=green, ultra thick] (2.5,0.5)--(2.5,-0.5);
\draw[color=gray, ->] (1.25,0) -- (2.75,0);

\draw[fill =gray] (1.5,0.5) circle (2pt);
\draw[fill =gray] (1.5,-0.5) circle (2pt);
\draw[fill =gray] (2.5,0.5) circle (2pt);
\draw[fill =gray] (2.5,-0.5) circle (2pt);

\end{tikzpicture}
\end{center}

\STATE Transform each vertex into  a face with as many sides as its degree. (The preceding splitting of edges implicitly accomplishes this. Each of these faces has a boundary of alternating blue and red edges.) Denote this expanded graph as $\overline{\Gamma}$.
\begin{center}
\begin{tikzpicture}
\foreach \i in {0,60,...,300} 
{
\draw [color=gray](0,0)--(\i:1);
}
\draw [fill  =gray] (0,0) circle (2pt);

\foreach \i in {0,120,...,240}
{
\draw [color=green, thick](3,0) +(\i-30:0.5)--+(\i+30:0.5);
\draw [color=red, thick](3,0) +(\i+30:0.5)--+(\i+90:0.5);
}

\foreach \i in {0,60,...,300} 
{
\draw [color=gray](3,0)--+(\i:1);
\draw [color=blue](3,0) +(\i+30:0.5)--+(\i+atan{0.25}:1);
\draw [color=blue](3,0) +(\i+30:0.5)--+(\i+60-atan{0.25}:1);
\draw [fill  =gray] ((3,0) +(\i+30:0.5) circle (2pt);
}

\end{tikzpicture}
\end{center}
\STATE With every edge $e=(u,v)$, associate a link operator  $\overline{K}_e\in \{X_u  X_v, Y_u Y_v, Z_u  Z_v \}$  depending on the color of the edge.
\STATE The gauge group is  given by $\mc{G} = \langle \overline{K}_e \mid e\in E(\overline{\Gamma})\rangle$.
\end{algorithmic}
\end{algorithm}

Our presentation slightly differs from that of \cite{bombin10} with respect to step 2.


\begin{theorem}[\cite{bombin10}]
Let $\Gamma_2$ be  a 2-colex  embedded on a surface of genus $g$. The 
subsystem code derived from $\Gamma_2$ via Construction~\ref{proc:tsc-bombin} has
 the following parameters:
 \ben
 [[3|V(\Gamma_2)|,2g,2|V(\Gamma_2)|+2g-2,d\geq \ell^\ast]],\label{eq:tsc-bombin}
 \een
 where $\ell^\ast$ is the length of smallest nontrivial cycle in $\Gamma_2^\ast$.
\end{theorem}
The cost of the two-body measurements is reflected to some extent in the increased overhead for the
subsystem codes. 
Comparing with the parameters of the color codes, this construction uses three times as many qubits as 
the associated color code while at the same time encoding half the number of qubits. Our codes
offer a different tradeoff between the overhead and distance. 

\subsection{Subsystem codes from 3-valent hypergraphs}
In this section we review a  general construction for (topological) subsystem codes based on hypergraphs proposed in \cite{suchara10}.  A hypergraph $\Gamma_h$ is an ordered pair $(V,E)$, where $E\subseteq 2^V$ is a collection of 
subsets of $V$. The set $V$ is called the vertex set while $E$ is called the edge set. If all the 
elements of $E$ are subsets of size 2, then $\Gamma_h$ is a standard graph. Any element of $E$ whose size is greater than 2 is
called a hyperedge and its rank is its size. The rank of a hypergraph is the maximum rank of its 
edges. A hypergraph is said to be of degree $k$ if at every site $k$ edges   are incident on it.

A hypercycle in a hypergraph is a set of edges such that on every vertex in the support of these edges
an even number of edges are incident \footnote{ There are various other definitions of hypercycles, see 
for instance \cite{duke85} for an overview.}. Note that this definition of hypercycle includes the 
standard cycles consisting of rank-2 edges. A hypercycle is said to have trivial homology if we can
contract it to a point, by contracting its edges. Homological equivalence of cycles is somewhat more complicated 
than in standard graphs.

The following construction is due to \cite{suchara10}. Let $\Gamma_h$ be a hypergraph satisfying the following conditions: 
\begin{compactenum}[H1)]
\item $\Gamma_h$ has only rank-2 and rank-3 edges. 
\item Every vertex is trivalent.
\item Two edges intersect at most at one vertex\footnote{Condition H3 implies that the hypergraphs that we are interested are also reduced hypergraphs. A reduced hypergraph is one in which no edge is a subset of another edge. 
}. 
\item Two rank-3 edges are disjoint. 
\end{compactenum}

We assume that at every vertex there is a qubit.  For each rank-2 edge $e=(u,v)$  define a link 
operator $K_e$ where  
$K_e\in\{ X_u  X_v, Y_u  Y_v, Z_u   Z_v\}$ and for each rank-3 edge $(u,v,w)$ 
 define 
\ben K_e= Z_u  Z_v   Z_w.\label{eq:rank3LinkOp}
\een The assignment of these link operators is such that 
\ben
K_e K_{e'}= (-1)^{|e\cap e'|} K_{e'}K_e. \label{eq:commuteRelns}
\een
We denote the cycles of $\Gamma_h$  by $\Sigma_{\Gamma_h}$. 
Let  $\sigma$ be a hypercycle in $\Gamma_h$, then we associate a (cycle) operator $W(\sigma)$
to it as follows:
\begin{eqnarray}
W(\sigma)& = &\prod_{e\in \sigma} K_e\label{eq:loopOperator}.
\end{eqnarray}
The group of these cycle operators is denoted $\mc{L}_{\Gamma_h} $ and defined as 
\begin{eqnarray}
\mc{L}_{\Gamma_h} &=& \langle W(\sigma)\mid \sigma \mbox{ is a hypercycle in } \Gamma_h \rangle \label{eq:cycleGroup}
\end{eqnarray}
It is immediate that $\dim \mc{L}_{\Gamma_h} = \dim \Sigma_{\Gamma_h}$.

\begin{algorithm}[H]
 \floatname{algorithm}{Construction}
\caption{{\ensuremath{\mbox{ Topological subsystem code via hypergraphs}}}}\label{proc:tsc-suchara}
\begin{algorithmic}[1]
\REQUIRE {A hypergraph $\Gamma_h$ satisfying assumptions H1--4}
\ENSURE{A subsystem code specified by its gauge group $\mc{G}$. }
\STATE Color all the rank-3 edges, say with $r$. Then assign a 3-edge-coloring of $\Gamma_h$ using $\{r,g,b\}$.
\STATE Define a graph $\overline{\Gamma}$ whose vertex set is same as $\Gamma_h$.
\STATE For each rank-2 edge $(u,v)$ in $\Gamma_h$ assign an edge $(u,v)$ in $\overline{\Gamma}$ and
a link operator $\overline{K}_{u,v}=K_{u,v}$ as 
\be
\overline{K}_{u,v} =\left\{ \ba{cl} X_u X_v & (u,v) \text{ is }r\\
Y_u Y_v & (u,v) \text{ is }  g\\
Z_u Z_v & (u,v) \text{ is } b \ea\right.
\ee
\STATE For each rank-3 edge $(u,v,w)$ assign three edges in $\overline{\Gamma}$, namely, $(u,v), (v,w), (w,u)$ and  three link operators $\overline{K}_{u,v}=Z_uZ_v$,  $\overline{K}_{v,w}=Z_vZ_w$, and $\overline{K}_{w,u}=Z_wZ_u$.

\STATE Define the gauge group $\mc{G} = \langle \overline{K}_e \mid e\in \overline{\Gamma}\rangle $.
\end{algorithmic}
\end{algorithm}

\begin{theorem}[\cite{suchara10}]\label{th:suchara-Const}
A hypergraph $\Gamma$ satisfying the conditions H1-4, leads to a subsystem code whose gauge group is
the centralizer of $\Sigma_{\Gamma_h}$, i.e.,  $\mc{G} = C(\mc{L}_{\Gamma_h})$. 
\end{theorem}
Since $S=\mc{G}\cap C(\mc{G})$, a subgroup of cycles corresponds to the stabilizer. Let us denote
this subgroup of cycles  by $\Delta_{\Gamma_h}$. 
Note that we have slightly simplified the construction proposed in \cite{suchara10}, in that
we let our our link operators to be only $\{X\otimes X, Y\otimes Y, Z\otimes Z \}$. But we
expect that this results in no loss in performance, because the number of encoded qubits and the distance are topological invariants and are not affected by these choices. 

Our notation is slightly different
from that of \cite{suchara10}. We distinguish between the link operators associated with the 
hypergraph $\Gamma_h$ and the derived graph $\overline{\Gamma}_h$; they coincide for the rank-2 edges. 
Because the hypergraph is 3-edge-colorable, we can partition the edge set of the hypergraph as
$E(\Gamma_h) = E_r\cup E_g\cup E_b$ depending on the color. The derived graph $\overline{\Gamma}_h$
is not 3-edge-colorable, but we group its edges by the edges of the parent edges in $\Gamma_h$.
Thus we can partition the edges of $\overline{\Gamma}_h$ also in terms of color as 
$E(\overline{\Gamma}_h) = \overline{E}_r\cup \overline{E}_g\cup \overline{E}_b$.

This following result is a consequence of  the definitions of $\mc{G}$, $\Sigma_{\Gamma_h}$ and Theorem~\ref{th:suchara-Const}.  
\begin{corollary}\label{co:rank2Cycle}
If $\sigma$ is a cycle in $\Gamma_h$ and
consists of only rank-2 edges, then $W(\sigma)\in S$.
\end{corollary}
An obvious question posed by Theorem~\ref{th:suchara-Const} is how does one construct hypergraphs 
that satisfy these constraints. This question will occupy us in the next section.  A related question 
is the syndrome measurement schedule for the associated subsystem code. This will be addressed in 
Section~\ref{sec:decoding}.

\renewcommand{\thetheorem}{\arabic{theorem}}
\setcounter{theorem}{0}

\section{Proposed topological codes}\label{sec:tsc}

\subsection{Color codes}

While our main goal is to construct subsystem codes, our techniques use color codes as 
intermediate objects. The previously known methods \cite{bombin07b} for color codes do not exhaust all possible color codes. Therefore we make a brief digression to propose a new method to construct color codes.   Then we will return to the question of building subsystem codes. 

The constructions presented in this paper assume that the associated
graphs and hypergraphs are connected, have no loops and all embeddings are such that the faces are homeomorphic to unit discs, in other words, all our embeddings are 2-cell embeddings. 

\renewcommand{\thealgorithm}{\arabic{algorithm}}
\setcounter{algorithm}{0}

\renewcommand{\thealgorithm}{\arabic{algorithm}}

\begin{algorithm}[H]
 \floatname{algorithm}{Construction}
\caption{{\ensuremath{\mbox{ Topological color code construction}}}}\label{proc:tcc-new}
\begin{algorithmic}[1]
\REQUIRE {An arbitrary bipartite graph $\Gamma$.}
\ENSURE {A 2-colex $\Gamma_2$.}
\STATE Consider the embedding of the bipartite graph $\Gamma$ on some surface. Take the dual of $\Gamma$, denote it $\Gamma^\ast$.
\STATE Since $\Gamma$ is bicolorable, $\Gamma^\ast$ is a 2-face-colorable graph.
\STATE Replace every vertex of $\Gamma^\ast$ by a face with as many sides as its degree such that
every new vertex has degree 3.

\begin{center}
\begin{tikzpicture}
\foreach \i in {0,60,...,300} 
{
\draw [color=gray](0,0)--(\i+30:1);
}
\draw [fill  =gray] (0,0) circle (2pt);

\foreach \i in {0,120,...,240}
{
\draw [color=blue, ultra thick](3,0) +(\i-30:0.5)--+(\i+30:0.5);
\draw [color=red, ultra thick](3,0) +(\i+30:0.5)--+(\i+90:0.5);
}

\foreach \i in {0,60,...,300} 
{
\draw [color=green, ultra thick](3,0) +(\i+30:0.5)--+(\i+30:1);
\draw [color=black, fill  =gray] ((3,0) +(\i+30:0.5) circle (2pt);
}

\end{tikzpicture}
\end{center}
\STATE The resulting graph is a 2-colex. 

\end{algorithmic}
\end{algorithm}

\begin{theorem}[Color codes from bipartite graphs]\label{th:tcc-new}
Any 2-colex must be generated from Construction~\ref{proc:tcc-new} via  some bipartite graph.
\end{theorem}
\begin{proof}
Assume that there is a 2-colex that cannot be generated by Construction~\ref{proc:tcc-new}. Assuming that the faces and the edges are 3-colored using $\{r,g,b\}$, pick any color 
$c\in\{ r,g,b\}$. Then contract all the edges of the remaining colors, namely $\{r,g,b\}\setminus c$. This process shrinks the faces that are coloured $c$.
The $c$-colored faces become  the vertices of the resultant 2-face-colorable  graph.
 The dual of this graph is bipartite as only bipartite graphs are 2-colorable. But this is precisely the reverse of the process described above. Therefore,  the 2-colex must have risen from a bipartite graph. 
\end{proof}
Note that there need not be a unique bipartite graph that generates a color code. In fact,
three distinct bipartite graphs may generate the same color code, using the above construction.

We also note that the  2-colexes  obtained via construction~\ref{proc:tcc-bombin}  have the property that  for  one of the colours, all the faces are of size 4. 
The following result shows the relation between our result and Construction~\ref{proc:tcc-bombin}. The proof is straightforward and omitted.

\begin{corollary}\label{co:2valent}
The color codes arising from Construction~\ref{proc:tcc-bombin} can be obtained from Construction~\ref{proc:tcc-new} using bipartite graphs which have the property that one bipartition of vertices contains only vertices of degree two.
\end{corollary}

\subsection{Subsystem codes via color codes}

Here we outline a procedure to obtain a subsystem code from a color code. This uses the construction of
\cite{suchara10}. 
We first construct a hypergraph that satisfies H1--4. We start with a 2-colex that
has an additional restriction, namely it has a nonempty set of faces each of which has a doubly even 
number of vertices. 

\begin{algorithm}[H]
 \floatname{algorithm}{Construction}
\caption{{\ensuremath{\mbox{ Topological subsystem code construction}}}}\label{proc:tsc-new}
\begin{algorithmic}[1]
\REQUIRE {A 2-colex $\Gamma_2$, assumed to have a 2-cell embedding.}
\ENSURE {A topological subsystem code specified by the hypergraph $\Gamma_h$.}
\STATE We assume that the faces of $\Gamma_2$ are colored $r$, $b$, and $g$.
Let $\rm{F}_r$ be the collection of $r$-colored faces of $\Gamma_2$, and $\rm{F}\subseteq \rm{F}_r$ such 
that $|f|\equiv 0 \bmod 4$ and $|f|>4$ for all $f\in \rm{F}$. 
\FOR{$f \in \rm{F}$}
\STATE Add a face $f'$ inside $f$ such that
$|f|=2|f'|$. 
\STATE Take a collection of alternating edges in the boundary of $f$. These are $|f|/2$ in number
and are  all colored  either  $b$ or $g$. 
\STATE Promote them to rank-3 edges by adding a vertex from $f'$  so 
that  the resulting hyperedges do not ``cross'' each other.	In other words, the rank-3 edge is a 
triangle and the triangles  are disjoint. Two possible methods of inserting the rank-3 edges are 
illustrated in  Fig.~\ref{fig:insertRank3}. In the first method, the hyperedges can be inserted so that they are in the 
boundary of the  $g$ colored faces, see Fig.~\ref{fig:promotedFace-1}. 	 Alternatively, the hyperedges can be inserted so that they are in 
the boundary of the $b$ colored faces, see Fig.~\ref{fig:promotedFace-2}.

\STATE Color the rank-3 edge with the same color as the parent rank-2 edge.
\STATE Color the edges of $f'$  using  colors distinct from the color of the rank-3 edges incident on 
$f'$.
\ENDFOR 
\STATE Denote the resulting hypergraph $\Gamma_h$ and use it  to construct the 
 subsystem code as in Construction~\ref{proc:tsc-suchara}.
\end{algorithmic}
\end{algorithm}

\begin{figure}[htb]
 \centering
 \subfigure[A face $f$ in $\rm{F}$]{
\includegraphics{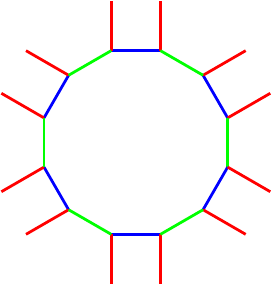}
    \label{fig:unpromotedFace}
}
   
    \subfigure[Inserting rank-3 edges in $f$ by promoting the $b$-edges to rank-3 edges.]{
\includegraphics{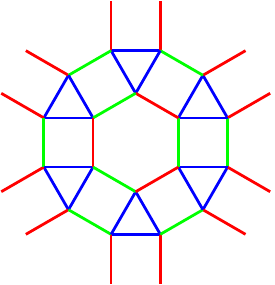}

   \label{fig:promotedFace-1}
   }
 \subfigure[Inserting rank-3 edges in $f$ by promoting the $g$-edges to rank-3 edges.]{
\includegraphics{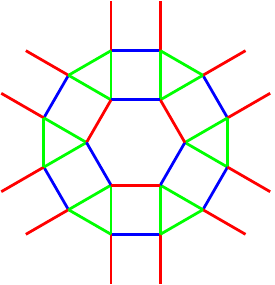}
   \label{fig:promotedFace-2}
   }
 \caption[Inserting rank-3 edges to convert $\Gamma_2$ to a hypergraph.]{%
(Color online) Inserting rank-3 edges in the faces of $\Gamma_2$ to obtain the hypergraph $\Gamma_h$. The rank-3 edges
correspond to triangles.}
 \label{fig:insertRank3}
\end{figure}

\begin{theorem}[Subsystem codes from color codes]\label{th:tsc-new}
Construction~\ref{proc:tsc-new} gives hypergraphs which satisfy the constraints H1-4 and therefore
give rise to 2-local  subsystem codes whose cycle group $\Sigma_{\Gamma_h} $ is defined as in Eq.~\eqref{eq:cycleGroup}
and gauge group is $\mc{G}=C(\mc{L}_{\Gamma_h})$.
\end{theorem}
\begin{proof}
Requirement H1 is satisfied because by construction, only rank-3 hyper edges are added to $\Gamma_2$,
which only contains rank-2 edges. 
The hypergraph has two types of vertices those that come from $\Gamma_2$ and those that are
added due to introduction of the hyperedges. Since all hyperedges come by promoting an edge to a 
hyperedge, it follows that the hypergraph is trivalent on the original vertices inherited from 
$\Gamma_2$. By construction, the vertices in $V(\Gamma_h)\setminus V(\Gamma_2)$ are trivalent and
thus $\Gamma_h$ satisfies H2. Note that $|f|\equiv0 \bmod 4$ and $|f|>4$, therefore $f'$ can be assigned an  edge coloring that ensures that $\Gamma_h$ is 3-edge colorable. Since $|f|>4$ we also ensure that no two edges intersect in more than one site, and H3 holds. By construction, all rank-3 edges are disjoint. This satisfies requirement H4.  \end{proof}

Let us illustrate this construction using a small example. 
It is based on the 2-colex shown in Fig.~\ref{fig:4-6-12lat}. The hypergraph derived
from this 2-colex is shown in Fig.~\ref{fig:4-6-12-hg}. 
Its rate is nonzero.

\begin{center}
\begin{figure}[htb]
\includegraphics[scale=0.5,angle=90]{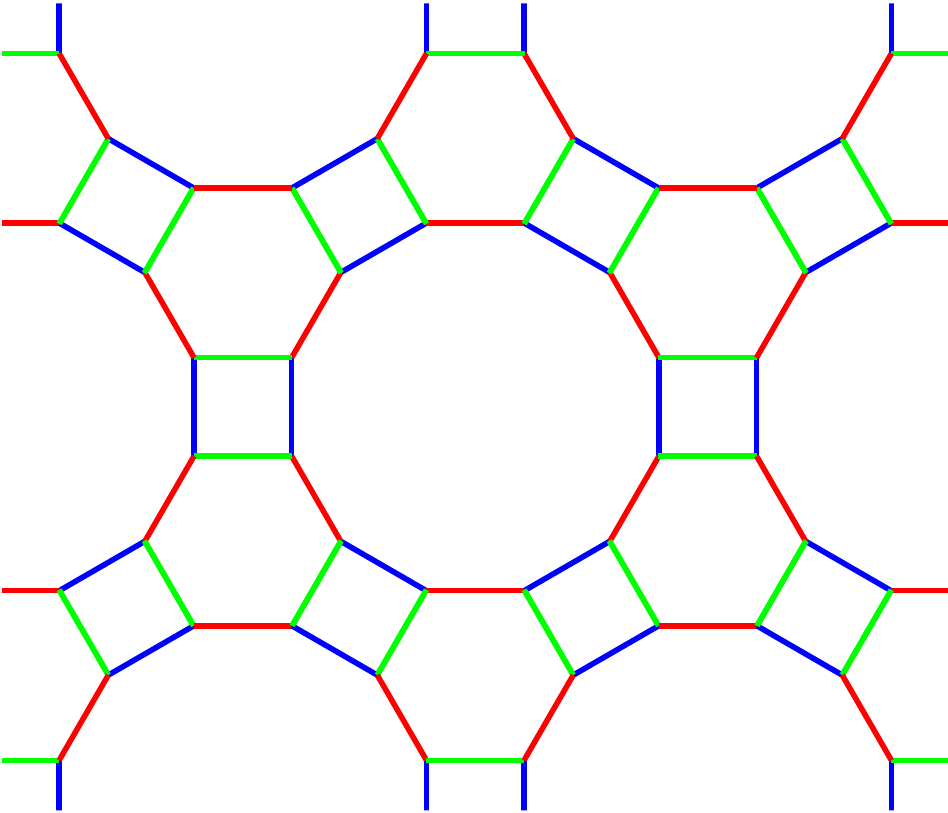}
\caption{(Color online) Color code on a torus from a 4-6-12 lattice. Opposite sides are identified.}\label{fig:4-6-12lat}
\end{figure}
\end{center}

\begin{center}
\begin{figure}[htb]
\includegraphics[scale=0.5,angle=90]{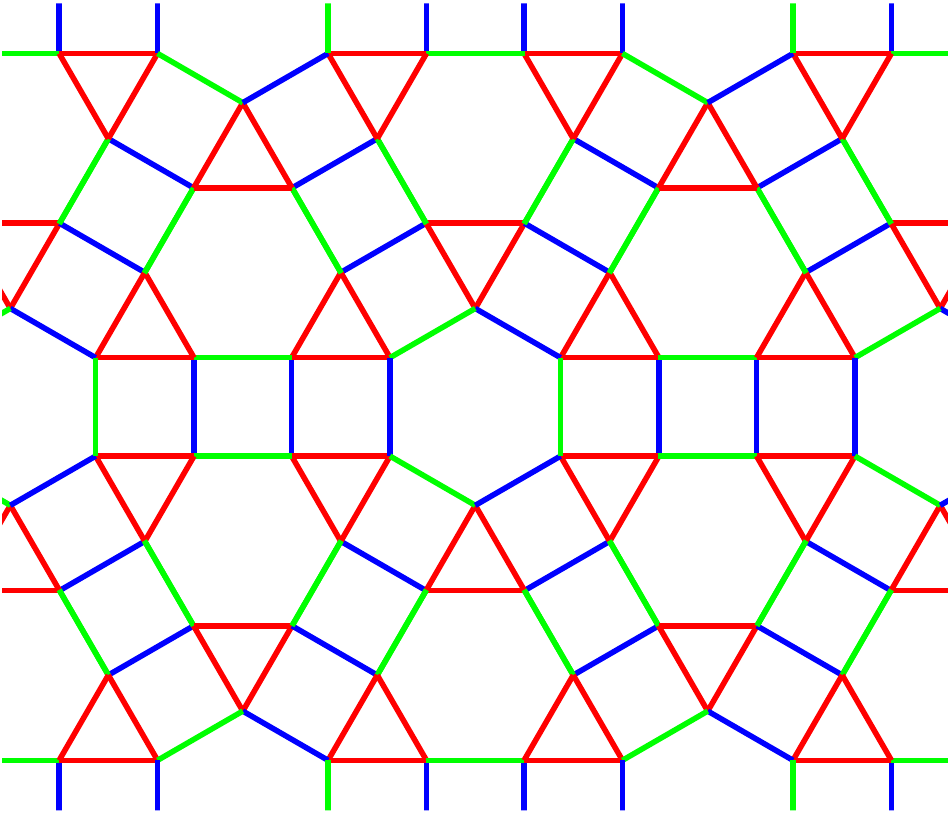}
\caption{(Color online) Illustrating  Construction~\ref{proc:tsc-new}.}\label{fig:4-6-12-hg}
\end{figure}
\end{center}

At this point, Theorem~\ref{th:tsc-new} is still 	quite general and we do not have expressions for
the code parameters in closed form. Neither is the structure of the stabilizer and the logical operators very apparent.  We impose some constraints on the set $\rm{F}$ so that we 
can remedy this situation. These restrictions still lead to a large class of subsystem codes. 
\begin{compactenum}[(i)]
\item $\rm{F}=\rm{F}_c$ is the set of all the faces 
of a given color; see Theorem~\ref{th:tsc-1}. 
\item $\rm{F}$  is an alternating set and $\rm{F}_c$ and $\rm{F}\setminus \rm{F}_c$ form a bipartite graph (in a sense which will be made precise shortly); see Theorem~\ref{th:tsc-2}. 
\end{compactenum}

Before, we can evaluate the parameters of these codes, we need some additional results
with respect to the structure of the stabilizer and the centralizer of the gauge group. 
The stabilizers vary depending on the set $\rm{F}$, nevertheless we can make some general statements
about a subset of these stabilizers. 

\begin{figure}[ht]
 \centering
 \subfigure[A hypercycle  $\sigma_1$ in $f$ (shown in bold edges) consisting of only rank-2 edges.]{
 \includegraphics{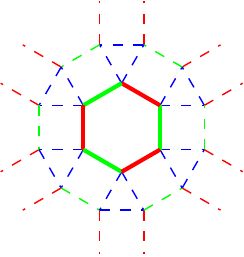}
   \label{fig:rank2Cycle}
   }
 \subfigure[A hypercycle $\sigma_2$  in $f$ (shown in bold edges) with both rank-2  and rank-3 edges.]{
 \includegraphics{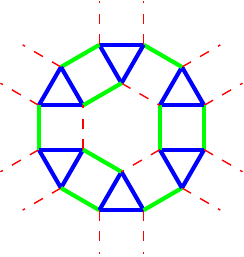}
   \label{fig:hyperCycle-1}
   }
 \subfigure[A dependent hypercycle 
 $\sigma_3$ which is a combination of $\sigma_1$ and $\sigma_2$ over $\F_2$.]{
 \includegraphics{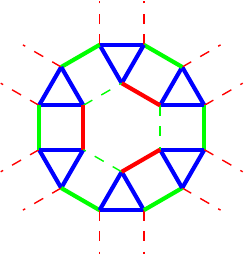}
   \label{fig:hyperCycle-2}
   }
\caption{(Color online) Stabilizer generators from a face in $\rm{F}$ for the subsystem codes of Construction~\ref{proc:tcc-new}; one of them is dependent. We shall view $\sigma_1$
and $\sigma_2$ as the two independent hypercycles associated with $f$.}\label{fig:stabGen-v-face-1}
\end{figure}

\begin{lemma}\label{lm:stabGens-tsc-1}
Suppose that $f$ is a face in $\rm{F}$ in Construction~\ref{proc:tcc-new}.
Then there are two independent hypercycles that we can associate with this face and consequently
two independent stabilizer generators as shown in Fig.~\ref{fig:stabGen-v-face-1}
\end{lemma}
\begin{proof}
We use the same notation as in Construction~\ref{proc:tcc-new}. Then Construction~\ref{proc:tcc-new} adds a new face $f'$ to $\Gamma_2$ in the interior of $f$. Let $\sigma_1$ be the 
cycle formed by the rank-2 edges in the boundary of $f'$, see Fig.~\ref{fig:rank2Cycle}.
 By Corollary~\ref{co:rank2Cycle},  $W(\sigma_1) \in S $.

Now let $\sigma_2$, see Fig.~\ref{fig:hyperCycle-1}, be the hypercycle consisting of all the edges in the boundary of $f$ and
an alternating set of rank-2 edges in the boundary of $f'$.  In other words, $\sigma_2$ consists of all the rank-3 edges inserted in  $f$ as well as the
rank-2 edges in its boundary and an alternating pair of rank-2 edges in $f'$. Because $|f|\equiv 0 \bmod 4$,
the boundary of $f'$ is 2-edge colorable.  
To prove  that $W(\sigma_2)$ can be generated by the elements of
$\mc{G}$, observe that $W(\sigma_2)$ can be split as 
\be
W(\sigma_2)  & = & \prod_{e\in \partial(f)} K_e \prod_{e\in \partial(f')\cap E_g} K_e,
\ee
where $E_g$ refers to the  $r$-colored edges  in $\Gamma_h$ and the boundary is with respect to 
$\Gamma_h$. We can also rewrite this in terms of the link operators 
in $\overline{\Gamma}_h$.
\be
W(\sigma_2) &= & \prod_{e\in \partial(f) } \overline{K}_e \prod_{e\in \partial(f') \cap \overline{E}_r}\overline{K}_e
\ee
where the boundary is with respect to 
$\overline{\Gamma}_h$ and $\overline{E}_r$ now refers to the $r$-colored edges in $\overline{\Gamma}_{h}$.

This is illustrated in Fig.~\ref{fig:hyperCycleDecompos}. 
The third cycle $\sigma_3$, see Fig.~\ref{fig:hyperCycle-2}, can be easily seen to be a combination of the cycles $\sigma_1$ and $\sigma_2$ over $\F_2$.
\end{proof}
\begin{figure}[ht]
 \centering
 \subfigure[Decomposing the hypercycle $\sigma_1$.]{
 \includegraphics{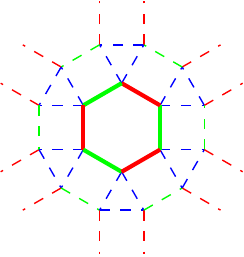}
   \label{fig:rank2CycleDecompos}
   }
 \subfigure[Decomposing the hypercycle $\sigma_2$.]{

\includegraphics{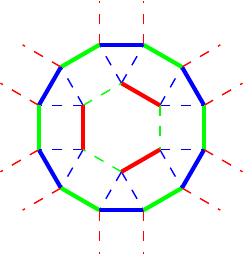}
   \label{fig:hyperCycleDecompos}
   }
\caption{(Color online) Decomposing  $\sigma_i$ in Fig.~\ref{fig:stabGen-v-face-1} so that $W(\sigma_i)$ 
 can be generated using the elements of $\mc{G}$. In each of the above $W(\sigma_i)$
 can be generated as the product of link operators corresponding to the bold edges.
 Note that these decompositions are with respect to the link operators of the derived graph $\overline{\Gamma}_h$ while
 the cycles are defined with respect to the hypergraph $\Gamma_h$.
 }\label{fig:stabGen-v-face-decompose}
\end{figure}

\begin{figure}[ht]
\centering
\includegraphics{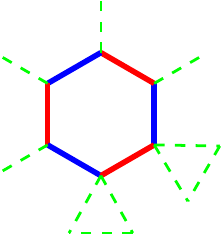}
\caption{(Color online) A  cycle $\sigma_1$ of rank-2 edges in the boundary of $f$, shown in bold, when $f$  has no rank-3 edges in its boundary. Some of the edges incident on $f$ maybe rank-3 but none in the boundary are.}
\label{fig:stabGen-f-face-1}
\end{figure}
\begin{figure}[ht]

\subfigure[A cycle $\sigma_1$ of rank-2 edges in the boundary of $f$, shown in bold,  
note that a rank-3 edge is incident on every vertex of $f$ unlike Fig.~\ref{fig:stabGen-f-face-1}.]{

\includegraphics{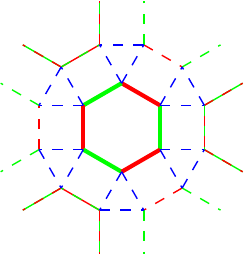}
	\label{fig:cycle-f-face-1}
}
\subfigure[A cycle $\sigma_2$ of rank-2  and rank-3, shown in bold; $\sigma_2$ differs from the cycle in Fig.~\ref{fig:hyperCycle-1}, in that the ``outer'' rank-2 edges maybe either $r$ or $g$.]{

\includegraphics{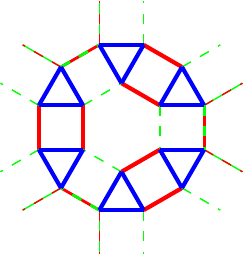}
	\label{fig:cycle-f-face-2}
}
\subfigure[Decomposing the hypercycle $\sigma_2$ so that $W(\sigma_2)$ 
 can be generated using the elements of $\mc{G}$. Note the decomposition refers to $\overline{\Gamma}_h$.]{

\includegraphics{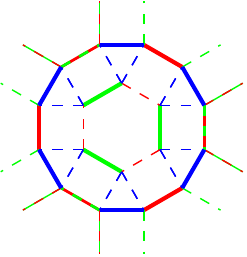}

	\label{fig:cycle-f-face-2-D}
}

\caption{ (Color online) Stabilizer generators for a face which has no rank-3 edges in its boundary when  $\rm{F}=\rm{F}_c $ and $f\not\in \rm{F}$.}
\label{fig:stabGen-f-face-2}
\end{figure}

\begin{lemma}\label{lm:stabGens-tsc-1-f}
Suppose that $f$ has no rank-3 edges in its boundary $\partial(f)$ as in Fig.~\ref{fig:stabGen-f-face-1}. Then  $W(\partial(f))$ is in $S$.
Further, if $\rm{F}=\rm{F}_r$ and $f\not\in \rm{F}$, then we can associate another hypercycle $\sigma_2$ to $f$, 
as in Fig.~\ref{fig:stabGen-f-face-2}, such that $W(\sigma_2)$ is in $S$. 
\end{lemma}
\begin{proof}
If $f$ has no rank-3 edges in its boundary, then $W(\partial(f))$ is in $S$ by Corollary~\ref{co:rank2Cycle}. It is possible that  some rank-3 edges are incident on $f$ even though they
are not in its boundary. This is illustrated in Fig.~\ref{fig:stabGen-f-face-1}.

If $\rm{F}=\rm{F}_r$, and $f\not\in \rm{F}$, then a rank-3 edge is incident on
every vertex of $f$ and we can form another cycle by considering all the rank-3 edges, and rank-2 edges
connecting all pairs of rank-3 edges, see Fig.~\ref{fig:cycle-f-face-2}.  This includes an alternating set of edges in the boundary of $f$. 
This is different  from the hypercycle in Fig.~\ref{fig:hyperCycle-1} in that the ``outer'' rank-2 
edges  connecting the rank-3 edges maybe of different color. 
Nonetheless by an augment similar to that in the proof of Lemma~\ref{lm:stabGens-tsc-1}, 
and using the decomposition shown in Fig.~\ref{fig:cycle-f-face-2-D} we can show that
 $W(\sigma_2)$ is in $S$. 
\end{proof}

\begin{remark}(Canonical cycles.) For the faces in which have two stabilizer generators associated
with them we  make the following canonical choice for the stabilizer generators. 
The first basis cycle $\sigma_1$ always refers to the cycle consisting of the rank-2 edges forming the 
boundary of a face. The second basis cycle for $f$ is chosen to be the cycle in which the rank-3 edges 
are paired with an adjacent rank-3 edge such that both the rank-2 edges pairing them are of the same 
color. 
\end{remark}

The decomposition as   illustrated in Fig.~\ref{fig:cycle-f-face-2-D} works even when the stabilizer 
is for a face which is adjacent to itself. 

Next we prove a bound on the distance of the codes obtained via Construction~\ref{proc:tsc-new}.
This is defined by the cycles in space $\Sigma_{\Gamma_h}\setminus \Delta_{\Gamma_h}$. Recall that
$W(\sigma)\in S$, if $\sigma\in \Delta_{\Gamma_h}$.

\begin{lemma}(Bound on distance)\label{lm:tsc-distance}
The distance of the subsystem code obtained from Construction~\ref{proc:tsc-new} is upper bounded by 
the number of  rank-3 edges in the hypercycle  with minimum number of rank-3  edges
in $\Sigma_{\Gamma_h}\setminus \Delta_{\Gamma_h}$. 
\end{lemma}
\begin{proof}
Every undetectable error of the subsystem code can be written as $gW(\sigma)$ for some $g\in \mc{G}$
and $\sigma\in \Sigma_{\Gamma_h}\setminus \Delta_{\Gamma_h}$. It suffices therefore, to check by how much the weight of 
$W(\sigma)$ can be reduced by acting with elements of $\mc{G}$.
In particular,  we can reduce $W(\sigma)$ such that only the 
rank-3 edges remain, and obtain an equivalent operator of lower weight. 
We can further act on this so that corresponding to every rank-3 edge in $\sigma$
the modified error has support only on one  of its vertices. 
This reduced error operator has weight equal to the number of rank-3 edges in $\sigma$. 
Thus the distance of the code is upper bounded by the number of rank-3 edges in the  hypercycle with minimum number of hyperedges in $\Sigma_{\Gamma_h}\setminus \Delta_{\Gamma_h}$. 
\end{proof}

It appears that this bound is tight, in that the distance is actually no less than the one specified
above. 

\begin{theorem}\label{th:tsc-1}
Suppose that $\Gamma$ is a graph such that every vertex has even degree greater than 2. Then construct the
2-colex $\Gamma_2$ from $\Gamma$ using Construction~\ref{proc:tcc-bombin}. Then apply Consruction~
\ref{proc:tsc-new} with $\rm{F}$ being  the set of $v$-faces of 
$\Gamma_2$  and with the rank-3 edges being in the boundaries of the $e$-faces of $\Gamma_2$. 
Let $\ell$ be the number of rank-3 edges in a cycle in $\Sigma_{\Gamma_h}\setminus \Delta_{\Gamma_h}$.
Then we obtain a 
\ben
[[6e,1+\delta_{\Gamma^\ast,\text{bipartite}}-\chi, 4e-\chi, d \leq \ell]]\label{eq:tscParams-1}
\een subsystem code
where $e=|E(\Gamma)| $  and $\delta_{\Gamma^\ast,\text{bipartite}}=1$  if $\Gamma^\ast$ is bipartite and zero otherwise.
\end{theorem}
\begin{proof}
Assume that $\Gamma$ has $v$ vertices, $f$ faces and
$e$ edges. Let us denote this by the tuple $(v,f,e)$, then  $\chi=v+f-e$. On applying Construction~\ref{proc:tcc-bombin}, we obtain a $2$-colex, $\Gamma_2$ with the parameters $(4e, v+f+e, 6e)$.
When we apply Construction~\ref{proc:tcc-new} to $\Gamma_2$, the resulting hypergraph $\Gamma_h$ has
$2e$ new vertices added to it. Further $2e$ edges are promoted to hyper edges, and as many new rank-2 edges
are created. Thus we have a hypergraph with $6e$ vertices, $2e$ hyperedges, $6e$ rank-2 edges.

The important thing to note is that the dimension of the hypercycle space of $\Gamma_h$ is related to 
$I_{\Gamma_{h}}$, the vertex-edge
incidence matrix of $\Gamma_h$. Let $E(\Gamma_h)$ denote the edges of $\Gamma_h$ including the
hyperedges. 
Then 
\ben
\dim \mc{L}_{\Gamma_h}  = |E(\Gamma_h)| - \rk_2(I_{\Gamma_h}), \label{eq:dimCycleSpace}
\een
where $\rk_2$ denotes the binary rank, \cite{duke85}.

 By Lemma~\ref{lm:rankH-1},
 $
\rk_2(I_{\Gamma_{h}})  = |V(\Gamma_h)|-1-\delta_{\Gamma^\ast,\text{bipartite}}.
$
It now follows that  
\be
\dim \mc{L}_{\Gamma_h}  &=&  |E(\Gamma_h)| - |V(\Gamma_h)| +1+\delta_{\Gamma^\ast,\text{bipartite}}.\\
&=& 8e - 6e+1+\delta_{\Gamma^\ast,\text{bipartite}}\\
&=& 2e+1+\delta_{\Gamma^\ast,\text{bipartite}}
\ee 

By Lemma~\ref{lm:stabGens-tsc-1}~and~\ref{lm:stabGens-tsc-1-f} every $v$-face and $f$-face of $\Gamma_2$ lead to two hypercycles in $\Gamma_h$. These are $2v+2f$ in number. 
But depending on whether $\Gamma^\ast$ is bipartite of these
only $s=2v+2f-1-\delta_{\Gamma^\ast,\text{bipartite}}$  are independent hypercycles. The dependencies
are as given below:
\ben
\prod_{f\in v\text{-faces}}  W(\sigma_1^f)& = & \prod_{f\in f\text{-faces}}  W(\sigma_2^f).\label{eq:vfaceDep0}
\een
If $\Gamma^\ast$ is bipartite then we have the following additional dependency. Let $\Gamma$ be
face-colored black and white so that 
$F(\Gamma) = F_1\cup F_2$, where $F_1$ and $F_2$ are the collection of black and white faces. Then 
\ben
\prod_{f\in f\text{-faces}} W(\sigma_1^f) \prod_{f\in  F_1} W(\sigma_2^f) & = & \prod_{f\in v\text{-faces}}  W(\sigma_2^f) \label{eq:vfaceDep1}\\
\prod_{f\in f\text{-faces}} W(\sigma_1^f) \prod_{f\in F_2} W(\sigma_2^f)&=& \prod_{f\in v\text{-faces}} W(\sigma_1^f)W(\sigma_2^f) \label{eq:vfaceDep2}
\een
(Note that among equations~\eqref{eq:vfaceDep0}--\eqref{eq:vfaceDep2} only two  are independent.)
All these are of trivial homology. There are no other independent cycles of trivial homology. 
Furthermore, Lemma~\ref{lm:nontrivialCycleProp-1}~and~\ref{lm:nontrivialCycleProp-2} show that hypercycles of nontrivial homology are not in the gauge group. Thus all the remaining (nontrivial) hypercycles are not in the stabilizer. 
We can now compute the number of encoded qubits as follows.
\be
2k&=&\dim C(\mc{G})-s\\
 & = & 2e+1+\delta_{\Gamma^\ast,\text{bipartite}} - (2v+2f-1-\delta_{\Gamma^\ast,\text{bipartite}})\\
& =& 2+2\delta_{\Gamma^\ast,\text{bipartite}}+2(e- v-f),
\ee
which gives  $k= 1+\delta_{\Gamma^\ast,\text{bipartite}}-\chi$  encoded qubits.
The number of gauge qubits $r$ can now be computed as follows:
\be
r &=& n-k-s\\
 & = &  6e - (1+\delta_{\Gamma^\ast,\text{bipartite}} -\chi) - (2v+2f-1-\delta_{\Gamma^\ast,\text{bipartite}})\\
& = & 6e-2v-2f = 4e-\chi.
\ee
The bound on distance follows from Lemma~\ref{lm:tsc-distance}.
\end{proof}

\begin{remark}
Note that there are no planar non-bipartite graphs $\Gamma^\ast$ which satisfy the constraint in 
Theorem~\ref{th:tsc-1}. 
\end{remark}

\begin{remark}
We might consider a variation is possible on the above, namely,  adding the hyper edges in the 
$f$-faces as opposed to the $v$-faces. This however does not lead to any new codes that are not constructible using Theorem~\ref{th:tsc-1}. Adding them in the $f$-faces is equivalent to applying Theorem~\ref{th:tsc-1} 
to the dual of $\Gamma$.
\end{remark}

In Theorem~\ref{th:tsc-1},  when $\Gamma^\ast$ is bipartite, the subsystem codes coincide with those obtained from \cite{bombin10}.
However in this situation,  a  different choice of  
$F$ in Construction~\ref{proc:tsc-new}  gives another family of codes that differ from \cite{bombin10}
and Theorem~\ref{th:tsc-1}. These codes are considered next. But first we need an intermediate result
about the hypercycles in $\Delta_{\Gamma_h}$ those that define the  stabilizer. Some of such as those in  Fig.~\ref{fig:stabGen-v-face-tsc-2-a}  are similar to those in  Fig.~\ref{fig:stabGen-v-face-1} but some such as those in 
Fig.~\ref{fig:stabGen-v-face-tsc-2-b} are not. 

\begin{figure}[htb]

 \subfigure[A hypercycle $\sigma_1$ for a $v$-face in $F$.]{
\includegraphics{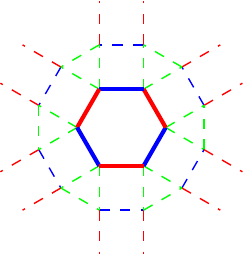}
   \label{fig:rank2Cycle-th2}
   }
\subfigure[A cycle $\sigma_2$ of rank-2  and rank-3 edges, shown in bold.]
{

\includegraphics{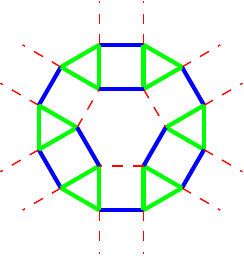}
   \label{fig:rank3Cycle-th2}
   }
\subfigure[Decomposing  $\sigma_2$ so that $W(\sigma_2)$ 
 can be generated using the elements of $\mc{G}$.]
 {

\includegraphics{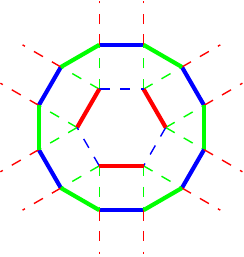}
   \label{fig:rank3Cycle-decomp}
   }

\caption{(Color online) Stabilizer generators for a $v$-face in $\rm{F}$, for the subsystem codes in Theorem~\ref{th:tsc-2}. Also shown is  the decomposition for  $W(\sigma_2)$. The decomposition for $W(\sigma_1)$ is same as in Fig.~\ref{fig:rank2Cycle-th2}.}
\label{fig:stabGen-v-face-tsc-2-a}
\end{figure}

\begin{figure}[htb]
\includegraphics{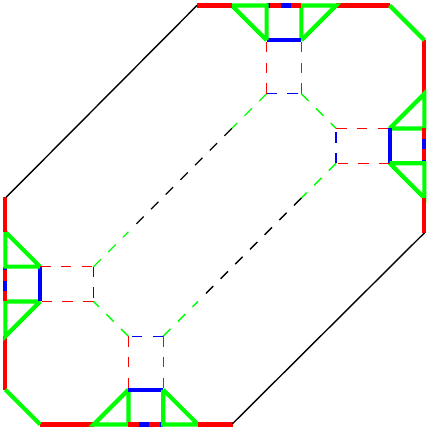}
\caption{(Color online) Stabilizer generators for a $v$-face in $\rm{F}_r\setminus \rm{F}$, for the subsystem codes in Theorem~\ref{th:tsc-2}. i) $\sigma_1= \partial(f)$ (not shown) and ii) $\sigma_2$ (in bold) consists of the rank-3 edges of all the adjacent $f$-faces in $\rm{F}$ adjacent through an $e$-face and the rank-2 edges connecting them.  The decomposition for $W(\sigma_2)$ is shown in Fig.~\ref{fig:stabGen-v-face-tsc-2-c}.} \label{fig:stabGen-v-face-tsc-2-b}
\end{figure}

Before, we give the next construction, we briefly recall the definition of a medial graph. The medial 
graph of a graph $\Gamma$ is obtained by placing a vertex on every edge of $\Gamma$ and an edge between
two vertices if and only if they these associated edges in $\Gamma$ are incident on the same vertex. 
We denote the medial graph of $\Gamma$ by $\Gamma_m$.

\begin{figure}[htb]
\includegraphics{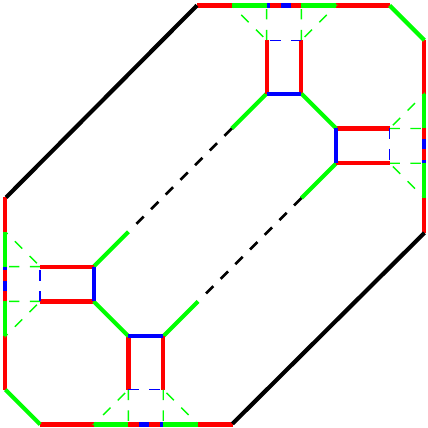}
\caption{(Color online)  Decomposition for  $W(\sigma_2)$. The product of the link operators shown in bold edges gives $W(\sigma_2)$. 
}
\label{fig:stabGen-v-face-tsc-2-c}

\end{figure}

\begin{theorem}\label{th:tsc-2}
Let $\Gamma$ be a graph whose vertices have even degrees greater than 2 and $\Gamma_m$ its medial graph. Construct the
2-colex $\Gamma_2$ from $\Gamma_m^\ast$ using Construction~\ref{proc:tcc-bombin}. 
Since $\Gamma_m^\ast$ is bipartite,  the set of v-faces of $\Gamma_2$, denoted $\rm{F}_r$, form a bipartition $\rm{F}_v\cup \rm{F}_f$, where $|\rm{F}_v| = |V(\Gamma)|$.
Apply Consruction~
\ref{proc:tsc-new} with the set $\rm{F}_v\subsetneq \rm{F}$ such that 
the rank-3 edges are not in the boundaries of the $e$-faces of  $\Gamma_2$.  
Let $\ell$ be the number of rank-3 edges in a cycle in $\Sigma_{\Gamma_h}\setminus \Delta_{\Gamma_h}$.
Then we obtain a 
\ben
[[10e,1-\chi+\delta_{\Gamma^\ast,\text{bipartite}}, 6e-\chi, d\leq\ell]]\label{eq:tscParams-2}
\een subsystem code, where $e=|E(\Gamma)|$. 
\end{theorem}
\begin{proof}
The proof is somewhat similar to that of Theorem~\ref{th:tsc-1}, but there are important differences.
Suppose that $\Gamma$ has $v$ vertices, $f$ faces and $e$ edges. Let us denote this as
the tuple $(v,f,e)$. The medial graph $\Gamma_m$ is 4-valent and has $e$ vertices, $v+f$ faces and $2e$
edges. The dual graph $\Gamma_m^\ast$ has the parameters $(v+f,e,2e)$. Furthermore, $\Gamma_m^\ast$ is bipartite. The 2-colex
$\Gamma_2$ has the parameters, $(8e,v+f+3e, 12e)$. Of the $v+f+3e$ faces 
$v+f$ are $v$-type, $e$ are $f$-type and $2e$ are $e$-type. 
The hypergraph has 
$10e$ vertices because a new vertex is added for every pair of rank-2 edge incident on the $v$-faces
in $\rm{F}_v$. These incident edges are all of one color, which are a third of the total edges of 
$\Gamma_m^\ast$ i.e., $(12e/3)$.  
Since a rank-3 edge is added only on one end of
these edges for every pair,  this implies that $2e$  
edges are promoted to rank-3 edges, as many new vertices and 
new rank-2 edges are added to form the hypergraph $\Gamma_h$.

By Lemma~\ref{lm:rankH-1}, the rank of the vertex-edge incidence matrix of $\Gamma_h$ is $|V(\Gamma_h)|-1-\delta_{\Gamma^\ast,\text{bipartite}} = 10e-1-\delta_{\Gamma^\ast,\text{bipartite}}$. The total number
of edges of $\Gamma_h$ is $14e$ including the rank-3 edges. Thus the rank of the cycle space of 
$\Gamma_h$ is 
\ben
\dim \mc{L}_{\Gamma_h}  &= &14e-10e+1+\delta_{\Gamma^\ast,\text{bipartite}}\\
&=&4e+1+\delta_{\Gamma^\ast,\text{bipartite}}.
\een

The stabilizer generators of this code are somewhat different than those in Theorem~\ref{th:tsc-1}.
Recall that  the $v$-faces form a bipartition, $\rm{F}_v\cup \rm{F}_f=\rm{F} \cup (\rm{F}_r\setminus \rm{F})$, where $|\rm{F}_v|=v$
and $|\rm{F}_f|=f$. We insert the rank-3 edges only in the faces in $\rm{F}$, and by Lemma~\ref{lm:stabGens-tsc-1} each of these faces leads to two stabilizer generators. These are illustrated in Fig.~\ref{fig:stabGen-v-face-tsc-2-a}. The remaining $v$-faces namely those in $\rm{F}_r\setminus \rm{F}$, have no rank-3 edges
in their boundary. Therefore, by Lemma~\ref{lm:stabGens-tsc-1-f} there is a stabilizer generator
associated with the boundary of the face. The other generator associated to a face in $\rm{F}_r\setminus \rm{F}$ is  slightly more complicated.  It is illustrated in Fig.~\ref{fig:stabGen-v-face-tsc-2-b}. The idea behind the decomposition so that it is an element of the gauge group is illustrated 
Fig.~\ref{fig:stabGen-v-face-tsc-2-c}.

Thus both the $v$-faces of $\Gamma_2$ give rise to two types of stabilizer generators. 
Since these are $v+f$ in number, we have $2(v+f)$ due to them. 
Each of the $e$-faces gives rise to one stabilizer generator giving $2e$ more generators. Thus there are 
totally $2(v+f)+2e$. However there are  some dependencies.
\ben
\prod_{f\in   \rm{F}_v} W(\sigma_1^f)&=& \prod_{f\in e\text{-faces}} W(\sigma_1^f)\prod_{f\in \rm{F}_f} W(\sigma_1^f)W(\sigma_2^f)
\een
When $\Gamma^\ast$ is bipartite, then it induces a bipartition on the 
$v$-faces in $\rm{F}_v=F_1\cup F_2$.  as well as the 
 $e$-faces, depending on whether the $e$-face is adjacent to a $v$-face in 
 $F_1$ or $F_2$. Denote this bipartition of $e$-faces as $E_1\cup E_2$.
 Then the following hold:
\be
\prod_{f\in   \rm{F}_v} W(\sigma_2^f) &=& \prod_{f\in  E_1}W(\sigma_1^f)\prod_{f\in  F_1}  W(\sigma_2^f)
\prod_{f\in  F_2}  W(\sigma_1^f)\\
\prod_{f\in  \rm{F}_v} W(\sigma_1^f) W(\sigma_2^f)&=& \prod_{f\in  E_2} W(\sigma_1^f)\prod_{f\in  F_1}  W(\sigma_1^f)\prod_{f\in  F_2}  W(\sigma_2^f)
\ee
Observe though there is only one new dependency when $\Gamma^\ast$ is bipartite.
The $f$-faces do not give rise to anymore independent 
generators. Thus there are $s=2(v+f+e)-1-\delta_{\Gamma^\ast,\text{bipartite}}$ independent
cycles of trivial homology. The remaining cycles are of nontrivial homology. By Lemma~\ref{lm:nontrivialCycleProp-1}~and~\ref{lm:nontrivialCycleProp-2}, these cycles are not in the gauge group.
Therefore the number of encoded qubits is given by 
\be
2k& =& \dim \mc{L}_{\Gamma_h}- s\\
&=& 4e+1+\delta_{\Gamma^\ast,\text{bipartite}}- 2(v+f+e)+1+\delta_{\Gamma^\ast,\text{bipartite}}\\
&=& 2+2\delta_{\Gamma^\ast,\text{bipartite}}+2(e-v-f)\\
&=& 2+2\delta_{\Gamma^\ast,\text{bipartite}}- 2\chi
\ee
Thus $k=1-\chi+\delta_{\Gamma^\ast,\text{bipartite}}$. It is now straightforward to compute the number of gauge qubits as 
$r=n-k-s = 10e-(1+\delta_{\Gamma^\ast,\text{bipartite}}-\chi)-2(v+f+e)+1+\delta_{\Gamma^\ast,\text{bipartite}} = 6e-\chi$.
The bound on distance follows from Lemma~\ref{lm:tsc-distance}.
\end{proof}

Theorem~\ref{th:tsc-2} can be strengthened without having to go through a medial graph but rather 
starting with an arbitrary graph $\Gamma$ and then constructing a 2-colex via Construction~\ref{proc:tcc-bombin}. 
We now demonstrate that  Construction~\ref{proc:tcc-new} gives rise to subsystem codes that 
are different from those obtained in \cite{bombin10}.
\begin{lemma}\label{lm:bombinHypergraphProperty}
Suppose that we have a topological subsystem code obtained by Construction~\ref{proc:tsc-bombin} 
from a  2-colex $\Gamma$. Then in the associated hypergraph shrinking 
the hyperedges to a vertex  gives a 6-valent graph and further replacing any multiple edges by a single edge gives us a 2-colex. 
\end{lemma}
\begin{proof}
Construction~\ref{proc:tsc-bombin} adds a rank-3 edge in every face of $\Gamma^\ast$. On contracting 
these rank-3 edges we end up with a  graph whose vertices coincide with the faces of $\Gamma^\ast$. 
Each of these vertices is now 6-valent and between any two adjacent vertices there are two edges. On 
replacing  these multiple edges by a single edge, we end up with a cubic graph. Observe that the 
vertices of this graph are in one to one correspondence with the faces of $\Gamma^\ast$ while the edges 
are also in one to one correspondence 
with the edges of $\Gamma^\ast$. Further an edge is present only if two faces are adjacent. This is 
precisely the definition of the dual graph. Therefore, the resulting graph is the same as $\Gamma$
and a 2-colex. 
\end{proof}

\begin{theorem}
Construction~\ref{proc:tsc-new} results in codes which cannot be constructed using Construction~\ref{proc:tsc-bombin}. In particular,  all the codes of Theorem~\ref{th:tsc-2} are distinct
from Construction~\ref{proc:tsc-bombin} and the codes of Theorem~\ref{th:tsc-1} are distinct
when $\Gamma$ therein  is non-bipartite.
\end{theorem}
\begin{proof}
Let us assume that the Construction~\ref{proc:tsc-new} does not give us {\em any} new codes. Then every 
code constructed using this method gives a code that is already constructed using 
Construction~\ref{proc:tsc-bombin}. Lemma~\ref{lm:bombinHypergraphProperty} informs us that
contracting the rank-3 edges results in a 6-valent graph, which on replacing the multiple edges by
single edge gives us a 2-colex. 

But note that if we applied the same procedure to a graph that is obtained from the proposed 
construction, then we do not always satisfy this criterion. In particular, this is the case for the
subsystem codes of Theorem~\ref{th:tsc-2}. These codes do not give rise to a 6-valent lattice on 
shrinking the rank-3 edges to a single vertex. 

When we consider the codes of Theorem~\ref{th:tsc-1}, on contracting that rank-3 edges, we end with up
a 6-valent graph with double edges and replacing them leads to a cubic graph. In order that these codes
do not arise from Construction~\ref{proc:tsc-bombin}, it is necessary that this cubic graph is not a 
2-colex. And if it were a 2-colex then further reducing the $v$-faces of this graph should give us a 
a 2-face-colorable graph. 
But this reduction results in  the graph we started out with namely, $\Gamma^\ast$. Thus when 
$\Gamma^\ast$ in non-bipartite, our codes are distinct from those in \cite{bombin10}.
\end{proof}

\begin{lemma}\label{lm:rankH-1}
The vertex-edge incidence matrices of the hypergraphs in Theorems~\ref{th:tsc-1}~and~\ref{th:tsc-2} have rank $|V(\Gamma_h)|-1-\delta_{\Gamma^\ast,\text{bipartite}}$.
\end{lemma}
\begin{proof}
We use the same notation as that of Theorems~\ref{th:tsc-1}~and~\ref{th:tsc-2}. Denote the vertex edge incidence matrix of 
$\Gamma_2$ as $I_{\Gamma_{2}}$. 
Depending on whether an edge in  $\Gamma_2$ is promoted to 
a  hyperedge in $\Gamma_h$ we can distinguish two types of edges in $\Gamma_2$. 
Suppose that the edges in $\{e_1,\ldots, e_l \}$ are not promoted while the edges in
$\{e_{l+1},\ldots, e_m \}$ are promoted.
\ben
I_{\Gamma_2} = \kbordermatrix{
    &e_1&\cdots&e_l&\vrule &e_{l+1}&\cdots& e_m \\ 
                &i_{11} &  \cdots  & i_{1l}&\vrule &\cdot &\cdots&i_{1m}\\
                & \vdots & \vdots & \ddots &\vrule& \vdots&\ddots&\vdots\\
                & i_{n1}  &   \cdots       &i_{nl }&\vrule&\cdot& \cdots& i_{nm}
}\label{eq:incidence2-colex}
\een
The vertex-edge incidence matrix of $\Gamma_{h}$ is related to that of $I_{\Gamma_2}$ as follows:
\ben
I_{\Gamma_h} &= &\kbordermatrix{
    &e_1&\cdots&e_l&\vrule &e_{l+1}&\cdots& e_m &\vrule&e_{m+1}&\cdots & e_{q}\\ 
                &i_{11} &  \cdots  & i_{1l}&\vrule &\cdot &\cdots&i_{1m}&\vrule&\\
                & \vdots & \vdots & \ddots &\vrule& \vdots&\ddots&\vdots&\vrule&&\bf{0}\\
                & i_{n1}  &   \cdots       &i_{nl }&\vrule&\cdot& \cdots& i_{nm}&\vrule&\\ 
                &   & \bf{0} &   &\vrule&  &\bf{I}&  &\vrule& &I_{\Gamma_h\setminus \Gamma_2}\\                
}\nonumber\\
& = & \left[\ba{ccc}\multicolumn{2}{c}{I_{\Gamma_2}}& 0 \\ 0 & I&I_{\Gamma_h\setminus \Gamma_2} \ea\right],\label{eq:incidenceHyper}
\een
where $I_{\Gamma_h\setminus \Gamma_2}$ is the incidence matrix of the subgraph obtained by restricting
to the vertices $ V(\Gamma_h)\setminus V(\Gamma_2)$. We already know that $\rk_2(I_{{\Gamma}_2})$
is $|V(\Gamma_2)|-1$.
Suppose there is an additional linear dependence among the rows of $I_{\Gamma_h}$.
More precisely, let  
\ben
b=\sum_{v\in V(\Gamma_2)} a_v \delta(v)  = 
\sum_{v\in V(\Gamma_h)\setminus V(\Gamma_2)} a_v \delta(v),\label{eq:b}
\een
where $\delta_{v}$ is the vertex-edge incidence vector of $v$.
Then $b$ must have no support on the edges in $\{e_1, \ldots, e_l \}\cup\{e_{m+1},\ldots, e_q \}$. It must have support only on the rank-3 edges of $\Gamma_h$.

Every rank-3 edge has the property that it is incident on exactly one vertex  $u\in V(\Gamma_h)\setminus (\Gamma_2)$ and exactly two vertices in 
$v,w\in V(\Gamma_2)$. Thus if a rank-3 edge has nonzero support in $b$, then $a_u\neq0$ and either
$a_v\neq 0$ or $a_w\neq0$ but not both.

\begin{center}

\begin{figure}[htb]	
\includegraphics{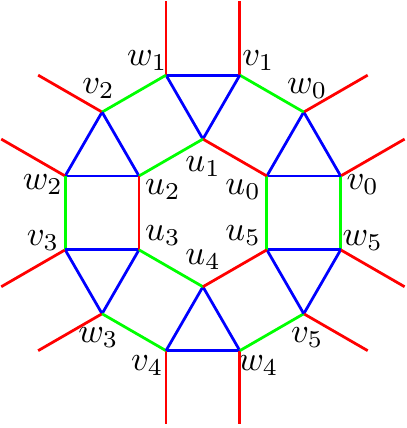}
\caption{ (Color online) If $b$ defined in Eq.~\eqref{eq:b} has support on one rank-3 edge of a $v$-face, then it has support on all the rank-3 edges of the $v$-face. Further, $\{a_{w_0}, a_{w_2}, a_{w_4}, \ldots \}\cup \{a_{v_1},a_{v_3},\ldots \}$ are all nonzero
or $\{a_{w_1}, a_{w_3}, \ldots \}\cup \{a_{v_0},a_{v_2},\ldots \}$ are all nonzero.}\label{fig:dependency-hg}
\end{figure}	
\end{center}

Suppose that a vertex $u_0 \in V(\Gamma_h)\setminus V(\Gamma_2)$ is such that $a_{u_0}\neq 0$.
Then because $b$ has no support on the edges in $\{e_{m+1},\ldots,e_{q} \}$, 
all the rank-2 neighbors of $u_0$, that is those which are connected by rank-2 edges are also such that $a_{u_i}\neq 0$. This implies that in a given $v$-face, for all the vertices of $u_i\in (V(\Gamma_h)\setminus V(\Gamma_2)) \cap f'$, we have $a_{u_i}\neq 0$. Further, only one of the rank-3 neighbors of $u_i$, namely $v_i,w_i$,  can have $a_{v_i}\neq 0 $ or $a_{w_i}\neq 0$, but not both. Additionally, pairs of these vertices must be adjacent as $b$ has no support on the rank-2 edges. 
Thus either $\{a_{w_0}, a_{w_2}, a_{w_4}, \ldots \}\cup \{a_{v_1},a_{v_3},\ldots \}$ are all nonzero
or $\{a_{w_1}, a_{w_3}, \ldots \}\cup \{a_{v_0},a_{v_2},\ldots \}$ are all nonzero.
Alternatively, we can say only the vertices in the support of an alternating set of rank-2 edges in the
boundary of the face can have nonzero $a_v$ in $b$. Consequently these vertices belong to an alternating
set of $f$-faces in the boundary of  $f$.

Consider now the construction in Theorem~\ref{th:tsc-1}, in this  rank-3 edges
are in the boundary of every $v$-face and $e$-face of $\Gamma_2$. Further, they are all connected.
Consider two adjacent $v$-faces as shown in Fig.~\ref{fig:dependency-hg-th1}.

\begin{figure}[htb]
\includegraphics{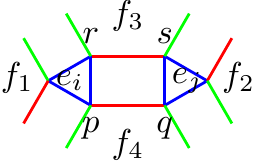}
\caption{(Color online) For the hypergraph in Theorem~\ref{th:tsc-1}, if $b$ has support on one rank-3 edge, then it has support on all rank-3 edges in
$\Gamma_h$. }\label{fig:dependency-hg-th1}
 \end{figure}
 
 If $ a_p\neq 0$, it implies that $a_r=0=a_s$ and $a_q\neq 0$.
If  the rank-3 edge  $e_j$ has support in $b$, then all the rank-3 edges incident on  $f_2$ must also be present. 
Since all the $v$-faces are connected, $b$ has support on all the rank-3 edges. 
Also note that the $f$-face $f_3$ has vertices in its boundary which are in the support of $b$.
In order that no edge from its boundary is in the support of $b$, all the vertices in its boundary
must be such that $a_v\neq 0$. The opposite holds for the vertices in $f_4$. None of these 
vertices must have $a_v\neq 0$. Thus the $f$-faces are portioned into two types and a consistent 
assignment of $a_v$ is possible if and only if the $f$-faces form a bipartition. In other words, 
$\Gamma^\ast$ is bipartite. Thus the additional linear dependency exists only when $\Gamma^\ast$
is bipartite.


Let us now consider the graph in Theorem~\ref{th:tsc-2}. In this case $\rm{F}$ and $\rm{F}_c\setminus \rm{F}$ form 
a bipartition. And only the
the set  $v$-faces in $F$ have the rank-3 edges in their boundary. Consider two adjacent $v$-faces of 
$\Gamma_2$,  $f_1\in \rm{F}_c\setminus \rm{F}$, $f_2 \in \rm{F}$, as shown in Fig.~\ref{fig:dependency-hg-th2}. 

\begin{figure}[htb]
\includegraphics{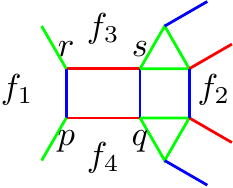}
\caption{(Color online) For the hypergraph in Theorem~\ref{th:tsc-2}, if $b$ has support on one rank-3 edge, then it has support on all rank-3 edges in
$\Gamma_h$. }\label{fig:dependency-hg-th2}
\end{figure}
In this case  $a_p=a_q=a_r=a_s$. So either all the vertices of $f_1$ are present
or none at all. This creates a bipartition of the $v$-faces which are not having rank-3 edges in their
boundary. Thus a consistent assignment of $a_v$ is possible if and only if the rest of the
$v$-faces in $F_c\setminus F $ form a bipartition. 
Since these are arising form the faces of $\Gamma$, this means that 
an additional linear dependency exists if and only if $\Gamma^\ast$ is bipartite. 
\end{proof}

\begin{lemma}\label{lm:nontrivialCycleProp-1}
Suppose that $\sigma$ is a homologically nontrivial hyper cycle of $\Gamma_h$ in Theorem~\ref{th:tsc-1}~or~\ref{th:tsc-2}. 
Then  $\sigma$ must contain some rank-3 edge(s). 
\end{lemma}
\begin{proof}
We use the notation as in Construction~\ref{proc:tsc-new}.
We can assume that such a cycle does not contain a vertex from $V(\Gamma_h)\setminus V(\Gamma_2)$. 
If such a vertex is part of the hyper cycle then all the vertices that belong to that $v$-face
are also part of it and there exists another cycle $\sigma'$ which consists of rank-2 edges and is not
incident on the vertices in $V(\Gamma_h)\setminus V(\Gamma_2)$. 

Suppose on the contrary that  $\sigma$ contains only rank-2 edges of $\Gamma_h$. In the hypergraphs
of Theorem~\ref{th:tsc-1}, every vertex in $\Gamma_h$  has one rank-3 edge incident on it, further each
vertex of $\Gamma_h$ is trivalent and 3-edge colourable with the rank-3 edges all colored the same. Therefore, 
$\sigma$ consists of rank-2 edges which are alternating in color. Every vertex is in the boundary of some $f$-face of $\Gamma_2$, say $\Delta$. Note that  an $f$-face does not have any rank-3 edge in its boundary although such an edge is incident on its vertices. This implies that $\sigma$ is the boundary of $\Delta$, therefore, homologically trivial cycle in contradiction our assumption. Therefore, $\sigma$ must contain some rank-3 edges. This proves the statement for the graphs  in Theorem~\ref{th:tsc-1}.

Suppose now that $\sigma$ is a cycle in the hypergraphs from Theorem~\ref{th:tsc-2}.
Now assume that there is a vertex 
in $\sigma$ that is in the $v$-face which has rank-3 edges in its boundary. This edge is incident 
on two vertices which are such that the rank-3 edges are in the boundary while the rank-2 edges are
out going and form the boundary of the 4-sided $e$-face incident on $u$, $v$. Therefore, the hyper cycle
$\sigma$ can be modified so that it is not incident on any $v$-face which has a rank-3 edge in its boundary. This implies from the $e$-faces only those edges are present in $\sigma$ that are in the
boundary of $e$-face and an $v$-face that has no rank-3 edges in its boundary. This edge is also coloured
same the color of the $f$-faces in $\Gamma_2$. Further $\sigma$ cannot have any edges that are
of the same color as the $v$-faces. Thus $\sigma$ must have the edges that are colored $b$ and
$g$ the colors of the $f$-faces and $e$-faces respectively. But this implies that $\sigma$ is the
union of the boundaries of $v$-faces, because only if there are edges of $r$-type can it leave the 
boundary of  a $v$-face. This contradicts that $\sigma$ is non trivial homologically. 
\end{proof}

\begin{lemma}\label{lm:nontrivialCycleProp-2}
Suppose that $\sigma$ is a homologically nontrivial hyper cycle of $\Gamma_h$ in Theorem~\ref{th:tsc-1} ~or~\ref{th:tsc-2}. Then $W(\sigma)$ is not in the gauge group.
\end{lemma}
\begin{proof}
Without loss of generality we can assume that $\sigma$ has a minimal number of rank-3 edges in it.
If not, we can compose it with another cycle in $\Delta_{\Gamma_h}$ to obtain one with fewer rank-3 edges. Note
that $W(\sigma) \in \mc{G}$ if and only if $W(\sigma')\in \mc{G}$.

Assume now that $W(\sigma)$ is in the gauge group. Let $E_2$ be the set of rank-2 edges and
$E_3$ be the set of rank-3 edges in $\Gamma_h$.
\be
W(\sigma) =\prod_{e\in E_2\cap \sigma}K_e\prod_{e\in E_3\cap\sigma} K_e 
\ee
The edges in $E_2\cap \sigma $ are also edges in $\overline{\Gamma}_h$ and the associated link
operators are the same. Therefore, it implies that $Z$-only operator
$ O_{\sigma} = \prod_{e\in E_3\cap\sigma} K_e$
is generated by the gauge group consisting of operators of the form $\{ X\otimes X,  Y\otimes Y, Z\otimes Z\}$. 

The operator $O_\sigma$ consists of (disjoint) rank-3 edges alone and therefore, 
for any edge $(u,v,w)$ in the support of $O_\sigma$, 
for each of the qubits $u,v,w$, one of the following must be true: 
\begin{compactenum}[(i)]
\item  Exactly one operator $Z_uZ_v$,
$Z_vZ_w$, $Z_wZ_u$ is required to generate the $Z_iZ_j$ on a pair of the qubits, where 
$i,j\in \{u,v,w\}$. The $Z$ operator on the 
remaining qubit is generated by gauge generators of the form $X_iX_j$ and $Y_iY_k$, where $i$ is one of
$\{u,v,w\}$ 
\item  The support on all the qubits is generated by 
$X_iX_j$ and $Y_iY_k$, where $i$ is one of $\{u,v,w\}$.
\end{compactenum}

For a qubit not in the support of $O_{\sigma}$, either no generator acts on it or all the three
gauge operators $X_uX_i$, $Y_u Y_j$, and $Z_u Z_v$ act on it. If it is the latter case, then it follows
that $u,v$ must be in the support of same rank-3 edge and that $v$ is also not in the support of 
$O_\sigma$.

Suppose that we can generate $O_{\sigma}$ as follows:
\be
O_{\sigma} &= & K^{(x,y)} K^{(z)},
\ee
where $K^{(x,y)}$ consists of only operators of the form $X\otimes X, Y\otimes Y$ and $K^{(z)}$
only of operators of the form $Z\otimes Z$. 
From the Lemma~\ref{lm:nontrivialCycleProp-1}, we see that the $O_{\sigma} K^{(z)}$
must be trivial homologically. The rank-3 edges incident on the
support of $O_{\sigma} K^{(z)}$ are either in the support of $\sigma$ or not.

 \begin{figure}[htb]
 \includegraphics{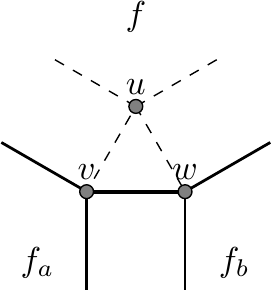}
 \caption{A rank-3 edge which is not in the support of $O_\sigma$. The solid edges indicate the link operators which are in the support of $K^{(x,y)}K^{(z)}$, while the dashed edges do not.
 The edge must occur in two cycles, one which encloses $f_a$, and another which encloses $f_b$. If the same
cycle encloses both $f_a$ and $f_b$, then the edge occurs twice in that cycle. If we consider the
stabilizer associated with these cycles then it has no support on this edge. 
 } \label{fig:nonOccuring}
 \end{figure}

A rank-3 edge $e$ which is not in 
$\sigma$ must be such that exactly two vertices from $e$ occur in the support of $O_{\sigma}K^{(z)}$.
There are two faces $f_a$ and $f_b$ associated \footnote{The two faces $f_a$ and $f_b$ could be the same face. In this case the associated cycle contains both the vertices.} with these two vertices, see Fig~\ref{fig:nonOccuring}. 
There is a hypercycle that encloses $f_a$ whose support contains $e$ and another that encloses $f_b$ 
and whose support  contains $e$. 
The product of these two stabilizer elements has no support on $e$ but has support on the
edges in $O_\sigma$. We can therefore, find an appropriate combination which are associated with the
trivial cycle in the support of $K^{(x,y)}$ such that $\sigma$ has fewer rank-3 edges. But this contradicts the minimality of rank-3 edges in $\sigma$. Therefore, it is not possible to 
 generate $W(\sigma)$ within the gauge group if $\sigma$ is homologically nontrivial. 
\end{proof}

{
\section{Syndrome measurement in topological subsystem codes}\label{sec:decoding}

One of the motivations for subsystem codes is the possibility of simpler recovery schemes. In this
section, we show how the many-body stabilizer generators can be measured using only  two-body
measurements. This could help in lowering the errors due to measurement and relax the error tolerance
on  measurement.

The proposed topological subsystem codes are not CSS-type unlike the Bacon-Shor codes. In CSS-type 
subsystem codes, the measurement of check operators is somewhat simple compared to the present case. 
The check operators are either $X$-type or $Z$-type. Suppose that we measure the $X$-type check 
operators  first. We can simply measure all the $X$-type gauge generators and then combine the outputs
classically to obtain the necessary stabilizer outcome. 
When we try to $Z$-type operator subsequently, we measure the $Z$-type gauge operators and once again 
combine them classically. This time around, the output of the $Z$-type gauge operators because they 
anti-commute with some $X$-type gauge operator we have uncertainty in the individual $Z$-type 
observables. Nonetheless because the $Z$-type check operator because it commutes with the $X$-type 
gauge generators, it can still be measured without inconsistency. 

When we deal with the non-CSS type subsystem codes, the situation is not so simple. We need to find
an appropriate decomposition of the stabilizers in terms of the gauge generators so that the individual
gauge outcomes can be consistently. So it must be demonstrated that the syndrome measurement can be performed by  measuring the gauge generators and a schedule exists for all the stabilizer generators. 
A condition that ensures that a certain decomposition of the stabilizer in terms of the 
gauge generators is consistent was shown in \cite{suchara10}.

\renewcommand{\thetheorem}{\Alph{theorem}}
\setcounter{theorem}{3}

\begin{theorem}[Syndrome measurement \cite{suchara10}]\label{lm:stabDecomp}
Suppose we have a decomposition of  a check operator $S$ as an ordered product of link operators $K_i$
such that 
\ben
S = K_m\cdots K_2 K_1 \mbox{  where } K_j \mbox{ is the link operator } K_{e_j}\\
\left[K_j, K_{j-1}\cdots K_1 \right] = 0 \mbox{ for all } j = 2,\cdots, m.
\een
Let  $s\in \F_2$ be the outcome of measuring S. Then to measure $S$,  measure the link operators  $K_i$ for $i=1$ to $m$
and compute $s=\oplus_{i=1}^{m} g_i$, where $g_i\in \F_2$ is the outcome of measuring $K_i$.
\end{theorem}

\renewcommand{\thetheorem}{\arabic{theorem}}
\setcounter{theorem}{8}

\renewcommand{\thetheorem}{\arabic{theorem}}
\setcounter{theorem}{8}

\begin{theorem}\label{th:syndrome}
The syndrome measurement of the subsystem codes in Theorems~\ref{th:tsc-1}~and~\ref{th:tsc-2} can be performed in three rounds using the following procedure, using the decompositions given in Fig.~\ref{fig:stabGen-v-face-decompose},~\ref{fig:stabGen-f-face-2}, for Theorem~\ref{th:tsc-1} and
Fig.~\ref{fig:stabGen-v-face-tsc-2-a},~\ref{fig:stabGen-v-face-tsc-2-c} for Theorem~\ref{th:tsc-2}.
\begin{compactenum}[(i)]
\item Let a stabilizer generator $W(\sigma) =\prod_{i} K_i \in S$ be decomposed as follows
\ben 
W(\sigma)=\prod_{i\in E_r}K_i \prod_{j\in E_g} K_j \prod_{k\in E_b} K_k = S_b S_g S_r 
\een where $K_i$ is a link operator and $E_r$, $E_g$, $E_b$ are the link operators corresponding to edges coloured $r$, $g$, $b$ respectively. 
\item Measure the gauge operators corresponding to the edges of different color at each level.
\item Combine the outcomes as per the decomposition of $W(\sigma)$.
\end{compactenum}
\end{theorem}
\begin{proof}
In the subsystem codes of Theorem~\ref{th:tsc-1}, there are two stabilizer generators associated with
the $v$-face and $f$-face. Those associated with the $v$-face are shown in 
Fig.~\ref{fig:stabGen-v-face-decompose}. Consider the first type of stabilizer generator 
$W(\sigma_1)$. Clearly, $W(\sigma_1)$ consists of two kinds of link operators, $r$ type and
$g$ type. The link operators  corresponding to the $r$-type edges are all disjoint and can therefore be measured in one round.  In the second round, we can measure the link operators corresponding to 
$g$-type edges. Since this is an even cycle we clearly have 
$[S_g,S_r]=0$. Note that $E_b=\emptyset$ because there are no $b$-edges in $\sigma_1$.
A similar reasoning holds for the generator $W(\sigma_1)$ shown in Fig.~\ref{fig:stabGen-f-face-2}
corresponding to an $f$-face.

For the second type of the stabilizer generators $W(\sigma_2)$, observe that as illustrated in
Fig.~\ref{fig:stabGen-v-face-decompose}, the $r$-edges are disjoint with the ``outer'' $b$ and $g$-edges
and can be measured in the first round. The ``outer'' $g$-edges being disjoint with the $r$-edges
we satisfy the condition $[S_g,S_r]=0$. In the last round when we measure the $b$-edges, since the
$b$-edges and $g$-edges overlap an even number of times and being disjoint with the $r$-edges
we have $[S_b,S_gS_r]=0$. Thus by Theorem~\ref{lm:stabDecomp}, this generator can be measured
by measuring the gauge operators. 

The same reasoning can be used to measure $W(\sigma_2)$ corresponding to the $f$-faces, but with
one difference. The outer edges are not all of the same color, however this does not pose a problem
because in this case as well we can easily verify that $[S_g,S_r]=0$, because they are
disjoint. Although the $b$-edges overlap with both the $r$ and $g$-edges note that each of them
individually commutes with $S_gS_r$ because they overlap exactly twice. Thus $[S_b,S_gS_r]=0$
 as well and we can measure $W(\sigma_2)$ through the gauge operators.

 Syndrome measurement of two disjoint stabilizers do not obviously interfere with each other. However,
 when two generators have overlapping support, they will not interfere as demonstrated below.  
 
 Note that every every vertex of $\Gamma_h $ in Theorem~\ref{th:tsc-1} has a rank-3 edge incident on it.
  As illustrated in Fig.~\ref{fig:nonIntereference}, edges which are not shared are essentially the rank-3 edges and each one of them figures in only one of the stabilizer generators, but because they 
  all commute they can be measured in the same round. The $r$ and $g$ edges are shared and
  appear in the support stabilizer generators of two adjacent faces. Nonetheless because edges of
  each color are disjoint they can be measured simultaneously. As has already been demonstrated the 
  edges of each color are such that for each stabilizer generator  $[S_g,S_r]=0$ and $[S_b,S_gS_r]=0$.

 \begin{figure}[htb]
\includegraphics{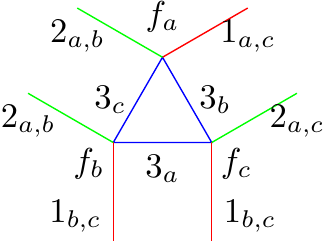}
 \caption{Noninterference of syndrome measurement. The faces $f_a$, $f_b$, $f_c$ have stabilizer generators that have overlapping support. The edges are labelled with the round in which they are
 measured, the subscripts indicate the faces with which the edge is associated. Thus $3_a$ indicates that this edge should be measured in the third round and it is used in the stabilizer generator $W(\sigma_2)$ of the face $f_a$.  } \label{fig:nonIntereference}
 \end{figure}
 A similar argument can be made for the codes in Theorem~\ref{th:tsc-2}, the proof is omitted.
 \end{proof}

The  argument  above  shows that the subsystem codes of Theorems~\ref{th:tsc-1}~and~\ref{th:tsc-2} can be measured in three rounds using the same procedure outlined in Theorem~
\ref{th:syndrome} if we assume that a single qubit can be involved in two distinct measurements
simultaneously. If this is not possible, then we need four time steps to measure all the checks. 
The additional time step is due to the fact that a rank-3 edge results in three link operators.
However only two of these are independent and they overlap on a single qubit. To measure both operators,  we need two 
time steps. Thus the overall time complexity is no more than four time steps. This is in contrast
to the schedule in \cite{bombin10}, which takes up to six time steps.

\section{Conclusion and Discussion}\label{sec:summary}


\subsection{Significance of proposed constructions}

To appreciate the usefulness of our results, it is helpful to understand Theorem~\ref{th:suchara-Const}
in more detail. 
First of all consider the complexity of finding hypergraphs which satisfy the requirements therein. 
Finding if a cubic graph is 3-edge-colorable is known to be NP-complete \cite{holyer81}. 
Thus determining if a 3-valent hypergraph is  3-edge-colorable is at least as hard. In view of the 
hardness of this problem, the usefulness of our results becomes apparent. 
One such family of codes is due to \cite{bombin10}. In this paper we provide new families of subsystem codes. 
Although they are also derived from  color codes, they lead to subsystem codes with different parameters. With respect to the results of \cite{suchara10}, our results bear a relation similar to a specific construction of quantum codes, say that of the Kitaev's topological codes, to the CSS construction.

Secondly, the parameters of the subsystem code constructed using Theorem~\ref{th:suchara-Const} depend on
the graph and the embedding of the graph. They are not immediately available in closed form for all
hypergraphs. We give two specific families of hypergraphs where the parameters can be given in closed form. In addition our class of hypergraphs naturally includes the hypergraphs arising in Bombin's 
construction. 

Thirdly,  Theorem~\ref{th:suchara-Const} 
does not distinguish between the case when the stabilizer is local and when the stabilizer
is non-local. Let us elaborate on this point. The subsystem code on the honeycomb lattice, for 
instance, can be viewed as a hypergraph albeit with no edges of rank-3. In the associated subsystem 
code the stabilizer can have support over a cycle which is nontrivial homologically. 
In fact, we can even provide examples of subsystem codes derived from true hypergraphs in that there 
exist edges or rank greater than two,  and whose stabilizer can  have elements associated to nontrivial 
cycles of the surface. Consider for instance, the 2-colex shown in Fig.~\ref{fig:4-8-colex}. 
The hypergraph derived from this 2-colex is shown in Fig.~\ref{fig:4-8-hg}. This particular code has a nonzero-rate even though its stabilizer includes cycles that are not nontrivial homologically. 

\begin{figure}[htb]
\subfigure[2-colex.]{
\includegraphics{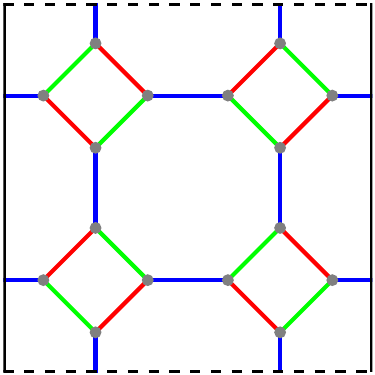}
\label{fig:4-8-colex}
}
\subfigure[Subsystem code.]{
\includegraphics{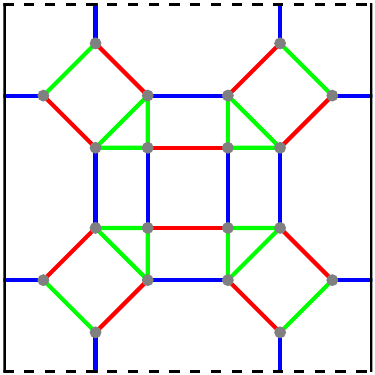}

\label{fig:4-8-hg}
}
\caption{(Color online) A subsystem code in which some of the stabilizer generators are nonlocal. This is derived from the color code on a torus from a square octagon lattice. Opposite sides are identified. }\label{fig:4-8-colex-2-hg}
\end{figure}

In contrast the subsystem codes proposed by Bombin, all have local stabilizers. It can be conceded that the locality of the stabilizer simplifies the decoding for stabilizer codes. But this is not necessarily  a restriction for the subsystem codes.
A case in point is the family of  Bacon-Shor codes which have non-local stabilizer generators. 
It would  be important to know what effect the non-locality of the stabilizer generators have on the threshold. 
Although we do not provide a criterion as to when the subsystem codes are topological in the sense
of having a local stabilizer, our constructions provide a partial answer in this direction. It would
be certainly more useful to have this criterion for all the codes of Theorem~\ref{th:suchara-Const}.

Not every cubic graph can allow us to define a subsystem code. Only if it satisfies the 
commutation relations, namely Eq.~\eqref{eq:commuteRelns} is it possible. As pointed out in 
\cite{suchara10}, the bipartiteness of the graph plays a role. The 
Petersen graph satisfies H1--4 being a cubic graph but with no 
hyperedges. But it does not admit a subsystem code to be defined because there is no consistent 
assignment of colours that enables the definition of the gauge group. In other words, we cannot assign 
the link operators such that Eq.~\eqref{eq:commuteRelns} are satisfied. We therefore, add the 3-edge-
colorability requirement to the hypergraph construction of the Suchara et al. \cite{suchara10}.

\begin{center}
\begin{figure}[htb]
\includegraphics{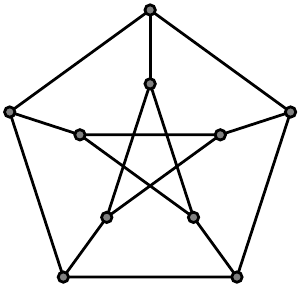}

\caption{The Petersen graph although cubic and satisfying H1--4, does not lead to a subsystem code via
Construction~\ref{proc:tsc-suchara}; it is  not 3-edge colorable.}
\end{figure}
\end{center}

Fig.\ref{fig:contrib} illustrates our contributions in relation to  previous work.

\begin{center}
\begin{figure}[htb]
\includegraphics{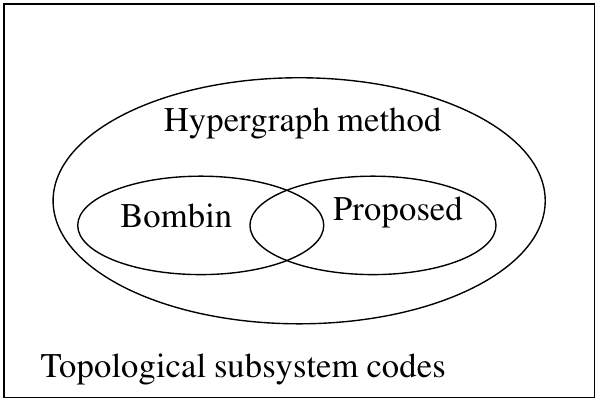}
\caption{Proposed constructions in context. Note that some of the hypergraph based subsystem codes may have homologically nontrivial stabilizer generators.
}\label{fig:contrib}
\end{figure}
\end{center}

\begin{acknowledgments}
 This work was supported by the Office of the Director of National Intelligence - Intelligence Advanced Research Projects Activity through Department of Interior contract D11PC20167. Disclaimer: The views and conclusions contained herein are those of the authors and should not be interpreted as necessarily representing the official policies or endorsements, either expressed or implied, of IARPA, DoI/NBC, or the U.S. Government.
\end{acknowledgments}

\def\cprime{$'$}

\end{document}